\documentclass[12pt]{article}
\usepackage{amsmath,amssymb,epsf,epsfig,amsthm,times,ifthen,epic}
\usepackage{enumerate}
\usepackage{cite}     
\usepackage{fancyheadings}
\usepackage{lastpage} 
\usepackage{hyperref} 

\def\submitteddate{\today}

\renewcommand{\baselinestretch}{1.0}
\begin{document}

\newcommand{\creationtime}{\today \ \ @ \theampmtime}

\pagestyle{fancy}
\renewcommand{\headrulewidth}{0cm}
\chead{\footnotesize{Dougherty-Freiling-Zeger}}
\rhead{\footnotesize{\submitteddate}}
\lhead{}
\cfoot{Page \arabic{page} of \pageref{LastPage}} 

\renewcommand{\qedsymbol}{$\blacksquare$} 


\newtheorem{theorem}              {Theorem}    
\newtheorem{lemma}      [theorem] {Lemma}
\newtheorem{corollary}  [theorem] {Corollary}
\newtheorem{proposition}[theorem] {Proposition}
\newtheorem{remark}     [theorem] {Remark}
\newtheorem{algorithm}  [theorem] {Algorithm}
\newtheorem{conjecture} [theorem] {Conjecture}
\newtheorem{example}    [theorem] {Example}

\theoremstyle{definition}
\newtheorem{definition} [theorem] {Definition}
\newtheorem*{claim}  {Claim}
\newtheorem*{notation}  {Notation}



\newcommand{\R}{\mathbf{R}}

\newcommand{\theranksymbol}{\mathsf{rank}}
\newcommand{\rank}[1]{\ \theranksymbol \left(#1\right)}
\newcommand{\nullity}[1]{\mathsf{N}(#1)}
\newcommand{\Dim}[1]{\mathsf{dim}(#1)}
\newcommand{\Comment}[1]{& [\mbox{from  #1}]}
\newcommand{\Pair}[2]{{#1},{#2}}
\newcommand{\Triple}[3]{{#1},{#2},{#3}}
\newcommand{\Quadruple}[4]{{#1},{#2},{#3},{#4}}

\newcommand{\filewebsite}{\url{http://zeger.us/linrank}}

\renewcommand{\emptyset}{\varnothing} 
\renewcommand{\subset}{\subseteq}     
\newcommand{\TBA}{*** To Be Added ***}

\begin{titlepage}

\setcounter{page}{0}

\title{Linear rank inequalities on five or more variables
\thanks{This work was supported by the Institute for Defense Analyses, the
National Science Foundation, and the UCSD Center for Wireless
Communications.\newline
\indent \textbf{R. Dougherty} is with the Center for Communications Research,
4320 Westerra Court, San Diego, CA 92121-1969 (rdough@ccrwest.org).\newline
\indent \textbf{C. Freiling} is with  the Department of Mathematics,
California State University, San Bernardino, 5500 University Parkway,
San Bernardino, CA 92407-2397 (cfreilin@csusb.edu).\newline
\indent \textbf{K. Zeger} is with the Department of Electrical and Computer
Engineering, University of California, San Diego, La Jolla, CA 92093-0407
(zeger@ucsd.edu).
}}

\author{Randall Dougherty, Chris Freiling, and Kenneth Zeger\\ \ }

\date{
\today \\
}

\maketitle
\begin{abstract}
Ranks of subspaces of vector spaces satisfy all linear inequalities
satisfied by entropies (including the standard Shannon inequalities)
and an additional inequality due to Ingleton.  It is known that the
Shannon and Ingleton inequalities generate all such linear rank
inequalities on up to four variables, but it has been an open
question whether additional inequalities hold for the case of five
or more variables.
Here
we give a list of 24 inequalities which, together with the
Shannon and Ingleton inequalities, generate all linear rank
inequalities on five variables. We also give a partial list of
linear rank inequalities on six variables and general results which
produce such inequalities on an arbitrary number of variables;
we prove that
there are essentially new inequalities at each number of
variables beyond four (a result also proved recently by Kinser).
\end{abstract}

%

\thispagestyle{empty}
\end{titlepage}

\clearpage


\section{Introduction}

It is well-known that the linear inequalities always satisfied by
ranks of
subspaces of a vector space (referred to here as \textit{linear
rank inequalities}) are closely related to the linear
inequalities satisfied by entropies of jointly distributed
random variables (often referred to as \textit{information
inequalities}).  For background material on this relationship
and other topics used here, a useful source is
Hammer, Romashchenko, Shen, and Vereshchagin
\cite{Hammer-Romashchenko-Shen-Vereshchagin97b}.

The present paper is about linear rank inequalities; nonetheless,
the basic results from information theory will be useful enough that
we choose to use the notation of information theory here.  We use
the following common definitions:
\begin{align*}
H(A|B) &= H(\Pair{A}{B}) - H(B) \\
I(A;B) &= H(A) + H(B) - H(\Pair{A}{B}) \\
I(A;B|C) &= H(\Pair{A}{C}) + H(\Pair{B}{C}) - H(\Triple{A}{B}{C}) -
H(C)
\end{align*}

There are two interpretations of these equations. When $A$, $B$, and
$C$ are random variables, $\Pair{A}{B}$ denotes the joint random
variable combining $A$ and $B$;  $H(A)$ is the entropy of $A$;
$H(A|B)$ is the entropy of $A$ given $B$; $I(A;B)$ is the mutual
information of $A$ and $B$; and $I(A;B|C)$ is the mutual information
of $A$ and $B$ given $C$.

But when $A$, $B$, and $C$ denote subspaces of a vector space, then
$\Pair{A}{B}$ denotes the space spanned by $A$ and $B$, which is
$\langle A,B\rangle$ or, since $A$ and $B$ are subspaces, just
$A+B$; $H(A)$ is the rank of $A$; $H(A|B)$ is the excess of the rank
of $A$ over that of $A \cap B$; $I(A;B)$ is the rank of $A \cap B$;
and $I(A;B|C)$ is the excess of the rank of $(A+C) \cap
(B+C)$ over that of $C$.  In either interpretation, the
equations above are valid.

The basic Shannon inequalities state that $I(A;B|C)$ (as well as the
reduced forms $I(A;B)$, $H(A|B)$, and $H(A)$) is nonnegative for any
random variables $A,B,C$.  Any nonnegative linear combination of
basic Shannon inequalities is called a Shannon inequality. We will
use standard Shannon computations such as
$I(A;B|C)=I(A;\Pair{B}{C})-I(A;C)$ (one can check this by expanding
into basic $H$ terms) and $H(A|C) \ge H(A|\Pair{B}{C})$ (because the
difference is $I(A;B|C)$) throughout this paper; an excellent source
for background material on this is Yeung~\cite{Yeung-book}.

A key well-known fact is that all information inequalities (and in
particular the Shannon inequalities) are also linear rank
inequalities for finite-dimensional vector spaces.  To see this,
first note that in the case of a \textit{finite} vector space $V$
over a finite field $F$, each subspace can be turned into a random
variable so that the entropy of the random variable is the same (up
to a constant factor) as the rank of the subspace: let $X$ be a
random variable ranging uniformly over $V^*$ (the set of linear
functions from $V$ to $F$), and to each subspace $A$ of $V$
associate the random variable $X\upharpoonright A$.  The entropy of
this random variable will be the rank of $A$, if entropy logarithms
are taken to base $|F|$.  For the infinite case,
one can use the theorem of Rado~\cite{Rado}
that any representable matroid is representable over a finite field,
and hence any configuration of finite-rank vector spaces over any
field has a corresponding configuration over some finite field.

The converse is not true; there are linear rank inequalities which
are not information inequalities.  The first such example is the
Ingleton inequality, which in terms of basic ranks or joint
entropies is
\begin{multline*}
     H(A) + H(B) + H(\Pair{C}{D}) + H(\Triple{A}{B}{C}) + H(\Triple{A}{B}{D}) \\
        \le H(\Pair{A}{B}) + H(\Pair{A}{C}) + H(\Pair{B}{C}) + H(\Pair{A}{D}) + H(\Pair{B}{D}),
\end{multline*}
but which can be written more succinctly using the $I$ notation as
\begin{equation*}
     I(A;B) \le I(A;B|C)+I(A;B|D)+I(C;D).
\end{equation*}
Ingleton~\cite{Ingleton} proved this inequality and asked whether there
are still further independent inequalities of this kind.

A key tool used by
Hammer et al.~\cite{Hammer-Romashchenko-Shen-Vereshchagin97b}
is the notion of common information.
A random variable $Z$ is a \textit{common information} of random
variables $A$ and $B$ if it satisfies the following conditions:
$H(Z|A)=0$, $H(Z|B)=0$, and $H(Z) = I(A;B)$.  In other words,
$Z$ encapsulates the mutual information of $A$ and $B$.  In general,
two random variables $A$ and $B$ might not have a common information.
But in the context of vector spaces (or the random variables
coming from them), common informations always exist; if $A$ and
$B$ are subspaces of a vector space, one can just let $Z$ be
the intersection of $A$ and $B$, and $Z$ will have the desired
properties.

Hammer et al.~\cite{Hammer-Romashchenko-Shen-Vereshchagin97b}
showed that the Ingleton inequality (and its permuted-variable
forms) and the Shannon inequalities fully characterize the
cone of linearly
representable entropy vectors on four random variables (i.e.,
there are no more linear rank inequalities to be found on
four variables).

\section{New five-variable inequalities}

We will answer Ingleton's question here. Using the existence of
common informations, one can prove the following twenty-four new
linear rank inequalities on five variables (this is a complete and
irreducible list, as will be explained below).
\newcommand{\Gap}{\ \ \ \ }
\begin{align}
     I(A;B) &\le I(A;B|C)+I(A;B|D)+I(C;D|E)+I(A;E) \\
     I(A;B) &\le I(A;B|C)+I(A;C|D)+I(A;D|E)+I(B;E) \\
     I(A;B) &\le I(A;C)+I(A;B|D)+I(B;E|C)+I(A;D|\Pair{C}{E}) \\
     I(A;B) &\le I(A;C)+I(A;B|\Pair{D}{E})+I(B;D|C)+I(A;E|\Pair{C}{D}) \\
     I(A;B) &\le I(A;C)+I(B;D|C)+I(A;E|D) \notag\\
                 &\Gap +I(A;B|\Pair{C}{E})+I(B;C|\Pair{D}{E}) \\
     I(A;B) &\le I(A;C)+I(B;D|E)+I(D;E|C) \notag\\
                 &\Gap +I(A;B|\Pair{C}{D})+I(A;C|\Pair{D}{E}) \\
     I(A;B) &\le I(A;C|D)+I(A;E|C)+I(B;D) \notag\\
                 &\Gap +I(B;D|\Pair{C}{E})+I(A;B|\Pair{D}{E}) \\
     2I(A;B) &\le I(A;B|C)+I(A;B|D)+I(A;B|E) \notag\\
                 &\Gap +I(C;D)+I(\Pair{C}{D};E) \\
     2I(A;B) &\le I(A;C)+I(A;B|D)+I(A;B|E) \notag\\
                 &\Gap +I(D;E)+I(B;\Pair{D}{E}|C)
\end{align}
\begin{align}
     2I(A;B) &\le I(A;B|C)+I(A;B|D)+I(C;D)+I(A;E) \notag\\
                  &\Gap +I(B;D|E)+ I(A;C|\Pair{D}{E}) \\
     I(A;\Pair{B}{C}) &\le I(A;C|\Pair{B}{D})+I(A;\Pair{C}{E})+I(A;B|\Pair{D}{E})+I(B;D|\Pair{C}{E}) \\
     I(A;\Pair{B}{C}) &\le I(A;C)+I(A;B|D)+I(A;D|E)+I(B;E|C) \notag\\
                  &\Gap +I(A;C|\Pair{B}{E})+ I(C;E|\Pair{B}{D}) \\
     I(A;\Pair{B}{C}) &\le I(A;B|D)+I(A;\Pair{C}{E})+I(B;D|\Pair{C}{E}) \notag\\
                  &\Gap +I(A;C|\Pair{B}{E})+I(C;E|\Pair{B}{D}) \\
     I(A;\Pair{B}{C}) &\le I(A;D)+I(B;E|D)+I(A;B|\Pair{C}{E}) \notag\\
                  &\Gap +I(A;C|\Pair{B}{D})+I(A;C|\Pair{D}{E}) \\
     I(A;\Pair{B}{C}) &\le I(A;D)+I(B;E|D)+I(A;C|E)+I(A;B|\Pair{C}{D}) \notag\\
                  &\Gap +I(A;C|\Pair{B}{D})+ I(B;D|\Pair{C}{E}) \\
     I(A;\Pair{B}{C}) &\le I(A;B|\Pair{C}{D})+I(A;C|\Pair{B}{D})+I(\Pair{B}{C};D|E) \notag\\
                  &\Gap +I(B;C|\Pair{D}{E})+I(A;E) \\
     I(\Pair{A}{B};\Pair{C}{D}) &\le I(\Pair{A}{B};D)+I(A;D|\Pair{B}{C})+I(B;D|\Pair{A}{C})+I(A;C|\Pair{B}{E}) \notag\\
                  &\Gap +I(B;C|\Pair{A}{E})+ I(A;B|\Pair{D}{E})+I(C;E|D) \\
     I(A;B)+I(A;C) &\le I(B;C)+I(A;B|D)+I(A;C|D)+I(B;D|E) \notag\\
                   &\Gap +I(C;D|E)+ I(A;E) \\
     I(A;B)+I(A;C) &\le I(B;D)+2I(A;C|D)+I(A;B|E)+I(D;E) \notag\\
                   &\Gap +I(B;E|\Pair{C}{D})+ I(C;D|\Pair{B}{E}) \\
     I(A;B)+I(A;C) &\le I(B;C)+I(B;D)+I(A;C|D)+I(A;B|E) \notag\\
                        &\Gap +I(A;E|B)+ I(C;D|E)+I(B;E|\Pair{C}{D}) \\
     I(A;B)+I(A;C) &\le I(B;D)+I(A;C|D)+I(A;D|E)+I(C;E) \notag\\
                         &\Gap +I(A;B|\Pair{C}{E})+ I(B;C|\Pair{D}{E})+I(B;E|\Pair{C}{D}) \\
     2I(A;B)+I(A;C) &\le I(A;B|C)+I(A;B|D)+I(C;D)+I(A;C|E) \notag\\
                         &\Gap +I(A;D|E)+2I(B;E)+I(B;C|\Pair{D}{E})+I(C;E|\Pair{B}{D}) \\
     I(A;B)+I(A;\Pair{B}{C}) &\le I(A;B|D)+2I(A;C|E)+I(B;E)+I(D;E) \notag\\
                         &\Gap +I(A;B|\Pair{C}{D})+2I(B;D|\Pair{C}{E})+I(C;E|\Pair{B}{D}) \\
     I(A;\Pair{C}{D})+I(B;\Pair{C}{D}) &\le I(B;D)+I(B;C|E)+I(C;E|D)+I(A;E)+I(A;C|\Pair{B}{D}) \notag\\
                          &\Gap +I(\Pair{A}{B};D|C)+I(A;D|\Pair{B}{E})+I(A;B|\Pair{D}{E})
\end{align}

(Note that there is much more variety of form in these inequalities
than there is in the four-variable
non-Shannon-type inequalities from~\cite{Dougherty-Freiling-Zeger-nonShannon}.)

Each of these inequalities is provable from the Shannon inequalities if we
assume that each mutual information on the left-hand side of the inequality
is in fact realized by a common information.  (Hence, since such common informations
always exist in the linear case, the inequalities are all linear rank
inequalities.)  For instance,
inequalities (1)--(10) all hold if we assume that there is a random variable $Z$
such that $H(Z|A)=H(Z|B)=0$ and $H(Z)=I(A;B)$; inequality (23) holds if there
exist random variables $Z$ and $Y$ such that $H(Z|A)=H(Z|B)=H(Y|A)=H(Y|\Pair{B}{C})=0$,
$H(Z)=I(A;B)$, and $H(Y)=I(A;\Pair{B}{C})$; and so on.  These assertions can all be
verified using the program {\tt ITIP}~\cite{ITIP}.
In fact, all of these become Shannon inequalities if we
replace the left-hand mutual information(s) with terms $H(Z)$ or $H(Y)$ and
add to the right-hand side appropriate terms like $kH(Z|A)+kH(Z|B)$ for a
sufficiently large coefficient $k$ ($k=5$ suffices for all of these
inequalities).  For example, for inequality (1), one can show that
\begin{align*}
    H(Z) &\le I(A;B|C)+I(A;B|D)+I(C;D|E)+I(A;E)+5H(Z|A)+5H(Z|B)
\end{align*}
is a Shannon inequality; if we set $Z$ to be a common information
for $A$ and $B$, we get inequality~(1).  Again the verifications of
these Shannon inequalities can be performed using {\tt ITIP}, or one can
work them out explicitly.
In Section~\ref{proofs} we will present
various alternate proof techniques.

These inequalities can be written in other equivalent forms.

Obvious rewrites (move the first term on the right to the left):
\begin{align}
  I(A;B|C) &\le I(A;B|D)+I(A;D|E)+I(B;E|C) \notag\\
                     &\Gap +I(A;C|\Pair{B}{E})+I(C;E|\Pair{B}{D}) \tag{12a}\\
  I(\Pair{A}{B};C|D) &\le I(A;D|\Pair{B}{C})+I(B;D|\Pair{A}{C})+I(A;C|\Pair{B}{E})\notag\\
                     &\Gap +I(B;C|\Pair{A}{E})+I(A;B|\Pair{D}{E})+I(C;E|D) \tag{17a}\\
  I(A;\Pair{C}{D})+I(B;C|D) &\le I(B;C|E)+I(C;E|D)+I(A;E)+I(A;C|\Pair{B}{D}) \notag\\
                            &\Gap +I(\Pair{A}{B};D|C)+I(A;D|\Pair{B}{E})+I(A;B|\Pair{D}{E}) \tag{24a}
\end{align}

Obvious rewrites (enlarge terms on the left so they can be combined):
\begin{align}
     2I(A;\Pair{B}{C}) &\le I(A;C|B)+I(A;B|C)+I(B;C)+I(A;B|D)+I(A;C|D) \notag\\
                   &\Gap +I(B;D|E)+I(C;D|E)+ I(A;E) \tag{18b} \\
     2I(A;\Pair{B}{C}) &\le I(A;C|B)+I(A;B|C)+I(B;D)+2I(A;C|D)+I(A;B|E) \notag\\
                   &\Gap +I(D;E)+I(B;E|\Pair{C}{D})+ I(C;D|\Pair{B}{E}) \tag{19b} \\
     2I(A;\Pair{B}{C}) &\le I(A;C|B)+I(A;B|C)+I(B;C)+I(B;D)+I(A;C|D) \notag\\
                        &\Gap +I(A;B|E)+I(A;E|B)+ I(C;D|E)+I(B;E|\Pair{C}{D}) \tag{20b} \\
     2I(A;\Pair{B}{C}) &\le I(A;C|B)+I(A;B|C)+I(B;D)+I(A;C|D)+I(A;D|E) \notag\\
                         &\Gap +I(C;E)+I(A;B|\Pair{C}{E})+ I(B;C|\Pair{D}{E})+I(B;E|\Pair{C}{D}) \tag{21b} \\
     3I(A;\Pair{B}{C}) &\le 2I(A;C|B)+2I(A;B|C)+I(A;B|D)+I(C;D)+I(A;C|E) \notag\\
                         &\Gap +I(A;D|E)+2I(B;E)+I(B;C|\Pair{D}{E})+I(C;E|\Pair{B}{D}) \tag{22b} \\
     2I(A;\Pair{B}{C}) &\le I(A;C|B)+I(A;B|D)+2I(A;C|E)+I(B;E)+I(D;E) \notag\\
                         &\Gap +I(A;B|\Pair{C}{D})+2I(B;D|\Pair{C}{E})+I(C;E|\Pair{B}{D}) \tag{23b} \\
     2I(\Pair{A}{B};\Pair{C}{D}) &\le I(B;\Pair{C}{D}|A)+I(A;\Pair{C}{D}|B)+I(B;D)+I(B;C|E) \notag\\
                          &\Gap +I(C;E|D)+I(A;E)+ I(A;C|\Pair{B}{D})+I(\Pair{A}{B};D|C) \notag\\
                          &\Gap +I(A;D|\Pair{B}{E})+I(A;B|\Pair{D}{E}) \tag{24b}
\end{align}

Non-obvious rewrites:
\begin{align}
   I(A;C) &\le I(A;C|B)+I(A;B|D)+I(C;D|E)+I(A;E) \tag{1c}\\
  I(A;B|C) &\le I(A;E|C)+I(A;C|\Pair{B}{D})+I(A;B|\Pair{D}{E})+I(B;D|\Pair{C}{E})\tag{11c}\\
  I(A;B|C) &\le I(A;B|D)+I(A;E|C)+I(B;D|\Pair{C}{E}) \notag\\
                     &\Gap +I(A;C|\Pair{B}{E})+I(C;E|\Pair{B}{D})\tag{13c}\\
  I(B;C|D) &\le I(B;C|\Pair{A}{D})+I(A;D|\Pair{B}{C})+I(B;E|D) \notag\\
                     &\Gap +I(A;C|E)+I(B;D|\Pair{C}{E})\tag{15c}\\
  I(B;C) &\le I(B;D)+I(A;C|D)+I(C;D|A) \notag\\
                   &\Gap +I(B;E|A)+I(B;C|\Pair{D}{E})+ I(D;E|\Pair{B}{C}) \tag{19c}\\
  I(C;D|E) &\le I(A;D|E)+I(C;D|A)+I(B;D|\Pair{C}{E}) \notag\\
                   &\Gap +I(B;\Pair{C}{E}|A)+I(C;E|\Pair{B}{D})\tag{21c}\\
  2I(A;\Pair{C}{D}) &\le I(A;D|C)+I(C;D|A)+I(A;C|B) \notag\\
                      &\Gap +I(A;D|B)+I(A;C|E)+I(A;D|E) \notag\\
                      &\Gap +2I(B;E)+I(B;C|\Pair{D}{E})+I(C;E|\Pair{B}{D}) \tag{22c}\\
  I(B;D|E) &\le I(B;D|A)+I(A;C|E)+I(C;E|A)+I(B;D|\Pair{A}{C}) \notag\\
                    &\Gap +I(D;E|\Pair{B}{C})+I(B;E|\Pair{C}{D})+I(B;D|\Pair{C}{E})\tag{23c}\\
  I(\Pair{A}{E};D) &\le I(B;D)+I(C;E|B)+I(D;E|C)+I(A;B|\Pair{C}{D}) \notag\\
                     &\Gap +I(A;D|\Pair{B}{C})+ I(A;D|\Pair{B}{E})+I(A;E|\Pair{B}{D})\tag{24c}
\end{align}

Note that, for these variant forms, we do \textit{not} make the claim
that the inequality follows from the existence of common informations
corresponding to the left-hand-side terms.  For instance, inequality~(19c)
does not follow from the Shannon inequalities and the existence of
a common information for $B$ and~$C$.  It turns out that inequality~(24b)
is provable from existence of a common information for $(A,B)$ and $(C,D)$, and
inequalities (19b), (21b), (22b), and~(23b) are provable from existence
of a common information for $A$ and~$(B,C)$, but inequalities (18b) and~(20b)
are not; in fact, no single common information (together with the
Shannon inequalities) suffices to prove (18) or (20).

\section{Alternate proofs and generalizations}\label{proofs}

In this section we will provide some alternate proof techniques for
the inequalities.  This will lead to natural generalizations.

\begin{lemma}
\label{lem:RtoZ}
The inequality
$H(Z|R)+I(R;S|T) \ge I(Z;S|T)$
is a Shannon inequality.
\end{lemma}

\begin{proof}
Using Shannon inequalities, we see that
\begin{align*}
 H(Z|R)+H(S|\Pair{Z}{T})
      &\ge H(Z|\Pair{R}{T})+H(S|\Pair{Z}{T})\\
      &\ge I(S;Z|\Pair{R}{T})+H(S|\Pair{Z}{T})\\
      &\ge I(S;Z|\Pair{R}{T})+H(S|\Triple{R}{Z}{T})\\
      &= H(S|\Pair{R}{T}).
\end{align*}
So $H(Z|R)-H(S|\Pair{R}{T}) \ge -H(S|\Pair{Z}{T})$; add $H(S|T)$ to both sides to
get the desired result.
\end{proof}

\begin{corollary}
If $H(Z|R)=0$, then $I(R;S|T)\ge I(Z;S|T)$.  \\
\label{cor:1}
\end{corollary}

\begin{proof}[Proof of the Ingleton inequality]
Let $Z$ be a common information of $A$ and $B$, so that
$H(Z|A)=H(Z|B)=0$ and $H(Z)=I(A;B)$.  Then
\begin{align*}
&I(A;B|C)+I(A;B|D)+I(C;D)\\
        &\ge I(Z;B|C)+I(Z;B|D)+I(C;D)  &\Comment{Corollary \ref{cor:1} using $H(Z|A)=0$}\\
      &\ge I(Z;Z|C)+I(Z;Z|D)+I(C;D)  &\Comment{Corollary \ref{cor:1} using $H(Z|B)=0$}\\
       &= H(Z|C)+H(Z|D)+I(C;D)\\
      &\ge H(Z|C) + I(Z;C)                   &\Comment{Lemma \ref{lem:RtoZ}}\\
      &\ge I(Z;Z)                    &\Comment{Lemma \ref{lem:RtoZ}}\\
       &= H(Z)\\
       &= I(A;B).
\end{align*}
\end{proof}

This is essentially the proof given in
Hammer et~al.~\cite{Hammer-Romashchenko-Shen-Vereshchagin97b}.

\begin{proof}[Proof of inequality (1)]
Let $Z$ be a common information of $A$ and $B$; then
\begin{align*}
&         I(A;B|C)+I(A;B|D)+I(C;D|E)+I(A;E)\\
      &\ge I(Z;Z|C)+I(Z;Z|D)+I(C;D|E)+I(Z;E)     &\Comment{Corollary \ref{cor:1} five times}\\
       &= H(Z|C)+H(Z|D)+I(C;D|E)+I(Z;E) \\
      &\ge I(Z;Z|E)+I(Z;E)                       &\Comment{Lemma \ref{lem:RtoZ} twice}\\
       &= H(Z|E)+I(Z;E) \\
       &= H(Z) \\
       &= I(A;B).
\end{align*}
\end{proof}

\begin{proof}[Proof of inequality (2)]
Let $Z$ be a common information of $A$ and $B$; then
\begin{align*}
&         I(A;B|C)+I(A;C|D)+I(A;D|E)+I(B;E) \\
      &\ge I(Z;Z|C)+I(Z;C|D)+I(Z;D|E)+I(Z;E)     &\Comment{Corollary \ref{cor:1}}\\
       &= H(Z|C)+I(Z;C|D)+I(Z;D|E)+I(Z;E) \\
      &\ge I(Z;Z|D)+I(Z;D|E)+I(Z;E)              &\Comment{Lemma \ref{lem:RtoZ}}\\
       &= H(Z|D)+I(Z;D|E)+I(Z;E)\\
      &\ge I(Z;Z|E)+I(Z;E)                       &\Comment{Lemma \ref{lem:RtoZ}}\\
       &= H(Z|E)+I(Z;E) \\
       &= H(Z)\\
       &= I(A;B).
\end{align*}
\end{proof}

The same pattern allows us to prove more general inequalities:
if $A_0$ and $B_0$ have a common information, then:
\begin{align}
      I(A_0;B_0) &\le I(A_0;B_0|B_1) \notag\\
                     &\Gap + I(A_0;B_1|B_2) \notag\\
                     &\Gap + \dotsb  \notag\\
                     &\Gap + I(A_0;B_{n-1}|B_n) \notag\\
                     &\Gap + I(B_0;B_n) \label{eq:starone} \\
      I(A_0;B_0) &\le 2^{n-1}I(A_0;B_0|A_1) + 2^{n-1}I(A_0;B_0|B_1)  \notag\\
               &\Gap + 2^{n-2}I(A_1;B_1|A_2) + 2^{n-2}I(A_1,;B_1|B_2) \notag\\
               &\Gap + \dotsb \notag\\
               &\Gap + I(A_{n-1};B_{n-1}|A_n) + I(A_{n-1};B_{n-1}|B_n) \notag\\
               &\Gap + I(A_n;B_n) \label{eq:startwo}
\end{align}
(Note that \eqref{eq:startwo} is related to results in Makarychev and Makarychev
\cite{Makarychev-Makarychev}.)
These can be generalized further; for instance, in the right hand side
of \eqref{eq:starone} any number of $A_0$'s may be replaced by $B_0$'s and/or
vice versa.

In fact:

\begin{theorem}
\label{thm:tree1} Suppose we have a finite binary tree where the
root is labeled with an information term $I(x;y)$ and each other
node is labeled with a term $I(x;y|z)$.
These terms may involve any variables.  We single
out two variables or combinations of variables, called $A$ and $B$.
Suppose that, for each node of the tree,
if its label is $I(x;y|z)$ [we allow $z$ to be empty at the root],
then:

$(a)$ $x$ is $A$ or $B$ and there is no left child, or

$(b)$ there is a left child and it is labeled $I(r;s|x)$ for some
$r$ and $s$;
\newline and

$(a')$ $y$ is $A$ or $B$ and there is no right child, or

$(b')$ there is a right child and it is labeled $I(r';s'|y)$ for
some $r',s'$.

Then the inequality
\begin{align}
    I(A;B) &\le \mbox{sum of all the node labels in the tree}
\end{align}
is a linear rank inequality (in fact, it is
true whenever $A$ and $B$ have a common information).
\end{theorem}

\begin{proof}
Let $Z$ be a new variable.
We prove by induction in the tree (from the leaves toward the root)
that, for each node $n$, if $T_n$ is the subtree rooted at $n$, and the node
label at $n$ is $I(r;s|t)$, then we have as a Shannon inequality
\begin{align}
    H(Z|t) &\le \mbox{sum of node labels in $T_n$}
              + j_n H(Z|A) + k_n H(Z|B)
\end{align}
for some $j_n,k_n \ge 0$.  (The inductive step uses Lemma~\ref{lem:RtoZ}.)
Applying this when $n$ is the root and $Z$ is a common information of $A$ and $B$
gives the desired result.
\end{proof}

We get the Ingleton inequality and inequalities (1) and (2) by applying this
to the trees:

\begin{center}
\begin{picture}(6.5,1)(-0.25,-0.25)
\thicklines
\put(1,0.5){\makebox(0,0)[r]{Ingleton:}}
\put(3,0.5){\makebox(0,0){$I(C;D)$}}
\put(2.625,0.125){\line(1,2){0.125}}
\put(3.375,0.125){\line(-1,2){0.125}}
\put(2.375,0){\makebox(0,0){$I(A;B|C)$}}
\put(3.625,0){\makebox(0,0){$I(A;B|D)$}}
\end{picture}
\end{center}

\begin{center}
\begin{picture}(6.5,1.5)(0.375,-0.25)
\thicklines
\put(1.625,1){\makebox(0,0)[r]{(1):}}
\put(3,1){\makebox(0,0){$I(A;E)$}}
\put(3.375,0.625){\line(-1,2){0.125}}
\put(3.625,0.5){\makebox(0,0){$I(C;D|E)$}}
\put(3.25,0.125){\line(1,2){0.125}}
\put(4,0.125){\line(-1,2){0.125}}
\put(3,0){\makebox(0,0){$I(A;B|C)$}}
\put(4.25,0){\makebox(0,0){$I(A;B|D)$}}
\end{picture}
\end{center}

\begin{center}
\begin{picture}(6.5,2)(0.6875,-0.25)
\thicklines
\put(1.9375,1.5){\makebox(0,0)[r]{(2):}}
\put(3,1.5){\makebox(0,0){$I(B;E)$}}
\put(3.375,1.125){\line(-1,2){0.125}}
\put(3.625,1){\makebox(0,0){$I(A;D|E)$}}
\put(4,0.625){\line(-1,2){0.125}}
\put(4.25,0.5){\makebox(0,0){$I(A;C|D)$}}
\put(4.625,0.125){\line(-1,2){0.125}}
\put(4.875,0){\makebox(0,0){$I(A;B|C)$}}
\end{picture}
\end{center}

A longer "linear" tree like the last one gives \eqref{eq:starone}, while a complete
binary tree of height $n$ gives \eqref{eq:startwo}.

Here is another version of Theorem~\ref{thm:tree1}:

\begin{theorem}
\label{thm:list1}
Let $I(x_1;y_1|w_1)$, $I(x_2;y_2|w_2)$, $\dotsc$, $I(x_m;y_m|w_m)$ be a
list of information terms, where each $x_i,y_i,w_i$ is chosen from the list
$A,B,r_1,r_2,\dots,r_k$ with the exception that $w_1$ is empty
(i.e., the first information term is just $I(x_1;y_1)$).
Suppose that each of the variables $r_j$ is used exactly twice,
once as a $w_i$ and once as an $x_i$ or $y_i$; while variables $A$ and $B$
may be used as many times as desired as an $x_i$ or $y_i$, but are not
used as a $w_i$.  Then the inequality
    $$I(A;B) \le \sum_{i=1}^m I(x_i;y_i|w_i)$$
is a linear rank inequality (in fact, it is
true whenever $A$ and $B$ have a common information).
\end{theorem}

\begin{proof}
We build a tree for use in
Theorem~\ref{thm:tree1}.  Each node will be labeled with one of the terms
$I(x_i;y_i|w_i)$.  The root is labeled $I(x_1;y_1)$.  If we have a
node $I(x_i;y_i|w_i)$ where $x_i$ is not $A$ or $B$, then create a
left child for this node and label it $I(x_j;y_j|w_j)$ for the
unique $j$ such that $w_j=x_i$.  Similarly, if $y_i$ is not $A$ or $B$,
then create a right child for this node and label it
$I(x_j;y_j|w_j)$ for the unique $j$ such that $w_j=y_i$.  It is easy to
show that no term $I(x_i;y_i|w_i)$ will be used more than once in
this construction (look for the counterexample nearest the root).
Hence, the constuction will terminate, and the sum of the labels
used is less than or equal to $\sum_{i=1}^m I(x_i;y_i|w_i)$ (it
does not matter if some of the terms $I(x_i;y_i|w_i)$ are not used
as labels).  Now Theorem~\ref{thm:tree1} gives the desired result.
\end{proof}

Theorem~\ref{thm:list1} directly gives the Ingleton inequality and
inequalities (1) and (2).  It also gives a number
of the other listed inequalities once we write them in an equivalent form
using equations such as $I(A;B|C)=I(A;\Pair{B}{C}|C)$:
\begin{align}
    I(A;B) &\le I(A;C)+I(A;B|D)+I(B;\Pair{C}{E}|C)+I(A;D|\Pair{C}{E})\tag{3d}\\
    I(A;B) &\le I(A;C)+I(A;B|\Pair{D}{E})+I(B;\Pair{C}{D}|C)+I(A;\Pair{D}{E}|\Pair{C}{D})\tag{4d}\\
    I(A;B) &\le I(A;C)+I(B;D|C)+I(A;\Pair{D}{E}|D) \notag\\
                   &\Gap +I(A;B|\Pair{C}{E})+I(B;\Pair{C}{E}|\Pair{D}{E})\tag{5d}\\
    I(A;B) &\le I(A;C|D)+I(A;\Pair{C}{E}|C)+I(B;D) \notag\\
                   &\Gap +I(B;\Pair{D}{E}|\Pair{C}{E})+I(A;B|\Pair{D}{E})\tag{7d}\\
   I(A;\Pair{B}{C}) &\le I(A;\Pair{B}{C}|\Pair{B}{D})+I(A;\Pair{C}{E})+I(A;\Pair{B}{D}|\Pair{D}{E}) \notag\\
                    &\Gap +I(\Pair{B}{C};\Pair{D}{E}|\Pair{C}{E})\tag{11d}\\
   I(A;\Pair{B}{C}) &\le I(A;C)+I(A;\Pair{B}{D}|D)+I(A;D|E) \notag\\
                   &\Gap +I(\Pair{B}{C};E|C)+I(A;\Pair{B}{C}|\Pair{B}{E})+ I(\Pair{B}{C};\Pair{B}{E}|\Pair{B}{D})\tag{12d}\\
   I(A;\Pair{B}{C}) &\le I(A;\Pair{B}{D}|D)+I(A;\Pair{C}{E})+I(\Pair{B}{C};D|\Pair{C}{E}) \notag\\
                   &\Gap +I(A;\Pair{B}{C}|\Pair{B}{E})+I(C;\Pair{B}{E}|\Pair{B}{D})\tag{13d}\\
   I(A;\Pair{B}{C}) &\le I(A;D)+I(\Pair{B}{D};\Pair{D}{E}|D)+I(A;\Pair{B}{C}|\Pair{C}{E}) \notag\\
                   &\Gap +I(A;\Pair{B}{C}|\Pair{B}{D})+ I(A;\Pair{C}{E}|\Pair{D}{E})\tag{14d}\\
   I(A;\Pair{B}{C}) &\le I(A;D)+I(\Pair{B}{D};E|D)+I(A;\Pair{C}{E}|E) \notag\\
                   &\Gap +I(A;\Pair{B}{C}|\Pair{C}{D})+I(A;\Pair{B}{C}|\Pair{B}{D})+I(\Pair{B}{C};\Pair{C}{D}|\Pair{C}{E})\tag{15d}\\
   I(A;\Pair{B}{C}) &\le I(A;\Pair{B}{C}|\Pair{C}{D})+I(A;\Pair{B}{C}|\Pair{B}{D})+I(\Pair{B}{C};\Pair{D}{E}|E) \notag\\
                   &\Gap +I(\Pair{B}{D};\Pair{C}{D}|\Pair{D}{E})+ I(A;E)\tag{16d}\\
   I(\Pair{A}{B};\Pair{C}{D}) &\le I(\Pair{A}{B};D)+I(\Pair{A}{B};\Pair{C}{D}|\Pair{B}{C})+I(\Pair{A}{B};\Pair{C}{D}|\Pair{A}{C}) \notag\\
                   &\Gap +I(\Pair{A}{B};\Pair{B}{C}|\Pair{B}{E})+I(\Pair{A}{B};\Pair{A}{C}|\Pair{A}{E})+I(\Pair{A}{E};\Pair{B}{E}|\Pair{D}{E}) \notag\\
                   &\Gap +I(\Pair{C}{D};\Pair{D}{E}|D)\tag{17d}
\end{align}
For instance, inequality (5d) is obtained from Theorem~\ref{thm:list1}
using the list of
random variables $$A,B,C,D,(\Pair{C}{E}),(\Pair{D}{E}).$$
Another approach is to prove the inequality
\begin{align*}
       I(A;B) &\le I(A;C)+I(B;D|C)+I(A;F|D)+I(A;B|E)+I(B;E|F)
\end{align*}
directly from Theorem~\ref{thm:list1} and then apply the variable substitution
\begin{align*}
       (A,B,C,D,E,F) \rightarrow (A,B,C,D,(\Pair{C}{E}),(\Pair{D}{E}))
\end{align*}
to get (5d).  Similarly, the other inequalities listed above are
substitution instances of linear-variable inequalities on five to eight
variables.  (Note that (3d), (4d), and (11d) are substitution instances of (1c).)

We will now generalize Theorem~\ref{thm:tree1} so as to generate additional
inequalities.  One easy but apparently useless generalization is to replace
the binary tree with a binary forest (a finite disjoint union of binary trees).
Then the hypotheses of Theorem~\ref{thm:tree1} can be stated just as
before (with ``the root'' replaced by ``each root''); and the conclusion is
the same except that the inequality becomes
\begin{align}
    mI(A;B) &\le \mbox{sum of all the node labels in the trees}
\end{align}
where $m$ is the number of trees (eqivalently, the number of root nodes).

This modification alone is useless because the resulting inequality is
just a sum of Theorem~\ref{thm:tree1} inequalities, one for each tree.
But it will become useful when combined with another modification.
For this we need a tightening of Lemma~\ref{lem:RtoZ}:

\begin{lemma}
\label{lem:RtoZplus}
The inequality
$H(Z|R)+I(R;S|T) \ge I(Z;S|T)+H(Z|\Triple{R}{S}{T})$
is a Shannon inequality.
\end{lemma}

\begin{proof} The proof is just as for Lemma~\ref{lem:RtoZ},
with the slack made explicit in one step.
Using Shannon inequalities, we see that
\begin{align*}
 H(Z|R)+H(S|\Pair{Z}{T})
      &\ge H(Z|\Pair{R}{T})+H(S|\Pair{Z}{T})\\
      &= H(Z|\Triple{R}{S}{T})+I(S;Z|\Pair{R}{T})+H(S|\Pair{Z}{T})\\
      &\ge H(Z|\Triple{R}{S}{T})+I(S;Z|\Pair{R}{T})+H(S|\Triple{R}{Z}{T})\\
      &= H(Z|\Triple{R}{S}{T})+H(S|\Pair{R}{T}).
\end{align*}
So $H(Z|R)-H(S|\Pair{R}{T}) \ge H(Z|\Triple{R}{S}{T})-H(S|\Pair{Z}{T})$;
add $H(S|T)$ to both sides to
get the desired result.
\end{proof}

Using this twice (and noting that $I(Z;Z|T)=H(Z|T)$ and
$H(Z|\Triple{Z}{S}{T})=0$), we get
\begin{align}
H(Z|R)+H(Z|S)+I(R;S|T) \ge H(Z|T)+H(Z|\Triple{R}{S}{T}).
\label{eq:ZRST}
\end{align}
The case where $T$ is a null variable gives
\begin{align}
H(Z|R)+H(Z|S)+I(R;S) \ge H(Z)+H(Z|\Pair{R}{S}).
\label{eq:ZRS}
\end{align}
These give us additional options in proving inequalities, as shown below.

\begin{proof}[Proof of inequality (8)]
Let $Z$ be a common information of $A$ and $B$; then
\begin{align*}
&         I(A;B|C)+I(A;B|D)+I(A;B|E)+I(C;D)+I(\Pair{C}{D};E)\\
      &\ge I(Z;Z|C)+I(Z;Z|D)+I(Z;Z|E)+I(C;D)+I(\Pair{C}{D};E) &\Comment{Corollary \ref{cor:1}}\\
       &= H(Z|C)+H(Z|D)+H(Z|E)+I(C;D)+I(\Pair{C}{D};E)\\
      &\ge H(Z)+H(Z|\Pair{C}{D})+H(Z|E)+I(\Pair{C}{D};E)               &\Comment{\eqref{eq:ZRS}}\\
      &\ge H(Z)+H(Z)+H(Z|\Triple{C}{D}{E})                        &\Comment{\eqref{eq:ZRS}}\\
      &\ge 2H(Z) \\
       &= 2I(A;B).
\end{align*}
\end{proof}

This proof immediately generalizes to give:
If $A$ and $B$ have a common information, then
\begin{align}
   (n-1)I(A;B) &\le I(A;B|C_1) + I(A;B|C_2) + \dots I(A;B|C_n) + {}\notag\\
          &\Gap  + [I(C_1;C_2) + I(C_1C_2;C_3) + \dots + I(C_1C_2 \dots C_{n-1};C_n)].
\label{eq:np2varineq}
\end{align}
The expression in brackets is actually symmetric in
$C_1,C_2,\dots,C_n$; it is equal to
\begin{align*}
H(C_1)+H(C_2)+\dots+H(C_n)-H(C_1C_2\dots C_n).
\end{align*}

One can use Lemma~\ref{lem:RtoZplus} to produce an extended form of
Theorem~\ref{thm:tree1}
in which an additional option is available: instead of having a
left child, a node can have a \textit{left pointer} pointing to some
other node anywhere in the tree or forest, and similarly on the right side.

\begin{theorem}
\label{thm:tree2}
Suppose we have a finite binary forest where each node is labeled
with an information term $I(x;y|z)$, where $z$ is empty at each root node
(i.e., the root labels are of the form $I(x;y)$).
These terms may involve any variables.  We single
out two variables or combinations of variables, called $A$ and $B$.
Suppose that, for each node of the
forest, if its label is $I(x;y|z)$ [with $z$ possibly empty],
then:

$(a)$ $x$ is $A$ or $B$ and there is no left child, or

$(b)$ there is a left child of this node and it is labeled
$I(r;s|x)$ for some $r,s$, or

$(c)$  there is a left pointer at this node pointing to some
other node whose label is $I(r';s'|t')$ where $x=(\Triple{r'}{s'}{t'})$;
\newline and

$(a')$ $y$ is $A$ or $B$ and there is no right child, or

$(b')$ there is a right child of this node and it is labeled
$I(r';s'|y)$ for some $r',s'$, or

$(c')$  there is a right pointer at this node pointing to some
other node whose label is $I(r';s'|t')$ where $y=(\Triple{r'}{s'}{t'})$.
\newline
Suppose further that no node is the destination of more than one pointer.
Let $m$ be the number of trees in the forest (equivalently, the number
of root nodes).  Then the inequality
\begin{align}
    mI(A;B) &\le \mbox{sum of all the node labels in the trees}
\label{eq:tree2ineq}
\end{align}
is a linear rank inequality (in fact, it is
true whenever $A$ and $B$ have a common information).
\end{theorem}

\begin{proof} As with Theorem~\ref{thm:tree1}, let $Z$ be a new variable.
For any left or right pointer, if $I(r;s|t)$ is the label at the destination
of the pointer, we say that the \textit{term associated with the pointer}
is $H(Z|\Triple{r}{s}{t})$.
We prove by induction in the forest (upward from the leaves toward the roots)
that, for each node $n$, if $T_n$ is the subtree rooted at $n$, and the node
label at $n$ is $I(r;s|t)$, then we have as a Shannon inequality
\begin{align}
\label{eq:forestsum}
    H(Z|t) &\le \mbox{sum of node labels in $T_n$} + \mbox{Out}_n - \mbox{In}_n
              + j_n H(Z|A) + k_n H(Z|B)
\end{align}
for some $j_n,k_n \ge 0$, where $\mbox{Out}_n$ is the sum of the terms associated
with pointers \textit{from} nodes in $T_n$ and $\mbox{In}_n$ is the sum of the
terms asasociated with pointers \textit{to} nodes in $T_n$.  (A pointer whose
source and destination are both in $T_n$ will contribute to both sums, but these
contributions will cancel each other out.) 
The inductive step uses Lemma~\ref{lem:RtoZplus}; the new term in that lemma
is used to handle the case where there is a pointer with destination~$n$
(note that, by assumption, there is at most one such pointer).
Once \eqref{eq:forestsum} is proved, apply it to all of the root nodes and
add the resulting inequalities together to get
\begin{align}
    mH(Z) &\le \mbox{sum of all the node labels in the trees}
              + j H(Z|A) + k H(Z|B)
\label{eq:tree2sum}
\end{align}
for some $j,k \ge 0$; the pointer sums cancel out because each pointer contributes
to one $\mbox{Out}$ sum and one $\mbox{In}$ sum.
Applying \eqref{eq:tree2sum} when $Z$ is a common information of $A$ and $B$
gives the desired result \eqref{eq:tree2ineq}.
\end{proof}

Theorem~\ref{thm:tree2}
can be used to prove inequalities (8) and (9)
using the following diagrams (pointers are represented as dashed curves):

\begin{center}
\begin{picture}(6.5,1.25)(-0.25,-0.25)
\thicklines
\put(0.4,0.5){\makebox(0,0)[r]{(8):}}
\put(4,0.5){\makebox(0,0){$I(\Pair{C}{D};E)$}}
\put(4.375,0.125){\line(-1,2){0.125}}
\put(2,0.5){\makebox(0,0){$I(C;D)$}}
\put(4.625,0){\makebox(0,0){$I(A;B|E)$}}
\put(1.625,0.125){\line(1,2){0.125}}
\put(2.375,0.125){\line(-1,2){0.125}}
\put(1.375,0){\makebox(0,0){$I(A;B|C)$}}
\put(2.625,0){\makebox(0,0){$I(A;B|D)$}}
\put(3.0,0.5){\makebox(0,0){\epsfxsize=1.5\unitlength\epsffile{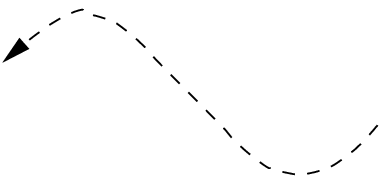}}}
\end{picture}
\end{center}

\begin{center}
\begin{picture}(6.5,1.5)(-0.25,-0.5)
\thicklines
\put(0.4,0.5){\makebox(0,0)[r]{(9):}}
\put(1.625,0.5){\makebox(0,0){$I(A;C)$}}
\put(2,0.125){\line(-1,2){0.125}}
\put(2.25,0){\makebox(0,0){$I(B;\Pair{D}{E}|C)$}}
\put(4.5,0.5){\makebox(0,0){$I(D;E)$}}
\put(4.125,0.125){\line(1,2){0.125}}
\put(4.875,0.125){\line(-1,2){0.125}}
\put(3.875,0){\makebox(0,0){$I(A;B|D)$}}
\put(5.125,0){\makebox(0,0){$I(A;B|E)$}}
\put(3.375,0.25){\makebox(0,0){\epsfxsize=1.8\unitlength\epsffile{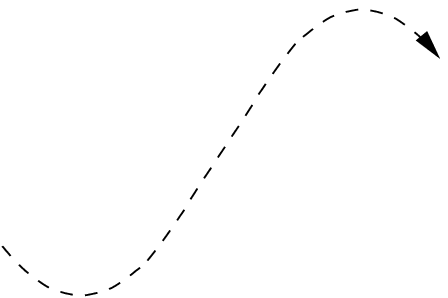}}}
\end{picture}
\end{center}

And by using equivalent forms of terms as was done in formulas (3d)
through~(17d), one can use Theorem~\ref{thm:tree2} to prove formulas
(6), (10), (19b), and (21b)--(24b) via the following diagrams:

\begin{center}
\begin{picture}(6.5,2.25)(-0.25,-0.5)
\thicklines
\put(0.4,1.5){\makebox(0,0)[r]{(6):}}
\put(2.125,1.5){\makebox(0,0){$I(A;C)$}}
\put(2.5,1.125){\line(-1,2){0.125}}
\put(2.75,1.0){\makebox(0,0){$I(\Pair{C}{D};E|C)$}}
\put(2.375,0.625){\line(1,2){0.125}}
\put(2.125,0.5){\makebox(0,0){$I(A;B|\Pair{C}{D})$}}
\put(3.125,0.625){\line(-1,2){0.125}}
\put(3.375,0.5){\makebox(0,0){$I(B;\Pair{D}{E}|E)$}}
\put(3.75,0.125){\line(-1,2){0.125}}
\put(4,0){\makebox(0,0){$I(A;\Triple{C}{D}{E}|\Pair{D}{E})$}}
\put(3.9375,0.5){\makebox(0,0){\epsfxsize=1.9\unitlength\epsffile{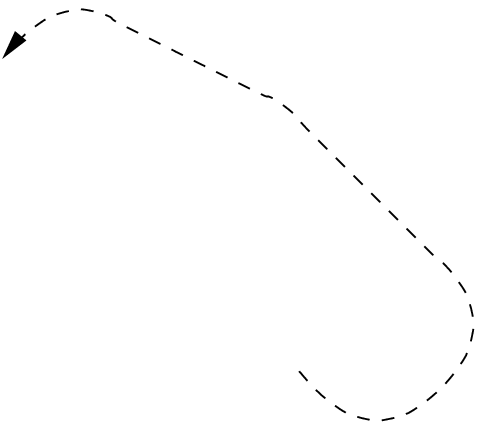}}}
\end{picture}
\end{center}

\begin{center}
\begin{picture}(6.5,2)(-0.25,-0.5)
\thicklines
\put(0.4,1.0){\makebox(0,0)[r]{(10):}}
\put(1.125,1.0){\makebox(0,0){$I(A;E)$}}
\put(1.5,0.625){\line(-1,2){0.125}}
\put(1.75,0.5){\makebox(0,0){$I(B;\Pair{D}{E}|E)$}}
\put(2.125,0.125){\line(-1,2){0.125}}
\put(2.375,0){\makebox(0,0){$I(A;\Pair{C}{D}|\Pair{D}{E})$}}
\put(4.125,1){\makebox(0,0){$I(C;D)$}}
\put(3.75,0.625){\line(1,2){0.125}}
\put(4.5,0.625){\line(-1,2){0.125}}
\put(3.5,0.5){\makebox(0,0){$I(A;B|C)$}}
\put(4.75,0.5){\makebox(0,0){$I(A;B|D)$}}
\put(3.25,0.5){\makebox(0,0){\epsfxsize=1.3\unitlength\epsffile{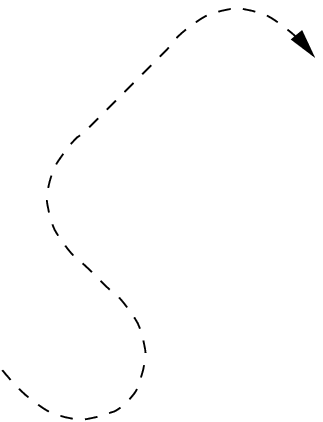}}}
\end{picture}
\end{center}

\begin{center}
\begin{picture}(6.5,2.5)(-0.25,-1)
\thicklines
\put(0.4,1){\makebox(0,0)[r]{(19b):}}
\put(1.125,1.0){\makebox(0,0){$I(B;D)$}}
\put(0.75,0.625){\line(1,2){0.125}}
\put(0.5,0.5){\makebox(0,0){$I(A;\Pair{B}{C}|B)$}}
\put(1.5,0.625){\line(-1,2){0.125}}
\put(1.75,0.5){\makebox(0,0){$I(A;\Pair{C}{D}|D)$}}
\put(2.125,0.125){\line(-1,2){0.125}}
\put(2.375,0){\makebox(0,0){$I(\Pair{B}{C};\Pair{D}{E}|\Pair{C}{D})$}}
\put(4.125,1){\makebox(0,0){$I(D;E)$}}
\put(3.625,0.625){\line(1,1){0.25}}
\put(4.5,0.625){\line(-1,2){0.125}}
\put(3.375,0.5){\makebox(0,0){$I(A;C|D)$}}
\put(3.75,0.125){\line(-1,2){0.125}}
\put(4,0){\makebox(0,0){$I(A;\Pair{B}{C}|C)$}}
\put(5.125,0.125){\line(-1,2){0.125}}
\put(5.375,0){\makebox(0,0){$I(\Pair{B}{C};\Pair{B}{D}|\Pair{B}{E})$}}
\put(4.75,0.5){\makebox(0,0){$I(A;\Pair{B}{E}|E)$}}
\put(3.25,0.5){\makebox(0,0){\epsfxsize=1.3\unitlength\epsffile{figs/pointer10.eps}}}
\put(3.59,0.25){\makebox(0,0){\epsfxsize=4.5\unitlength\epsffile{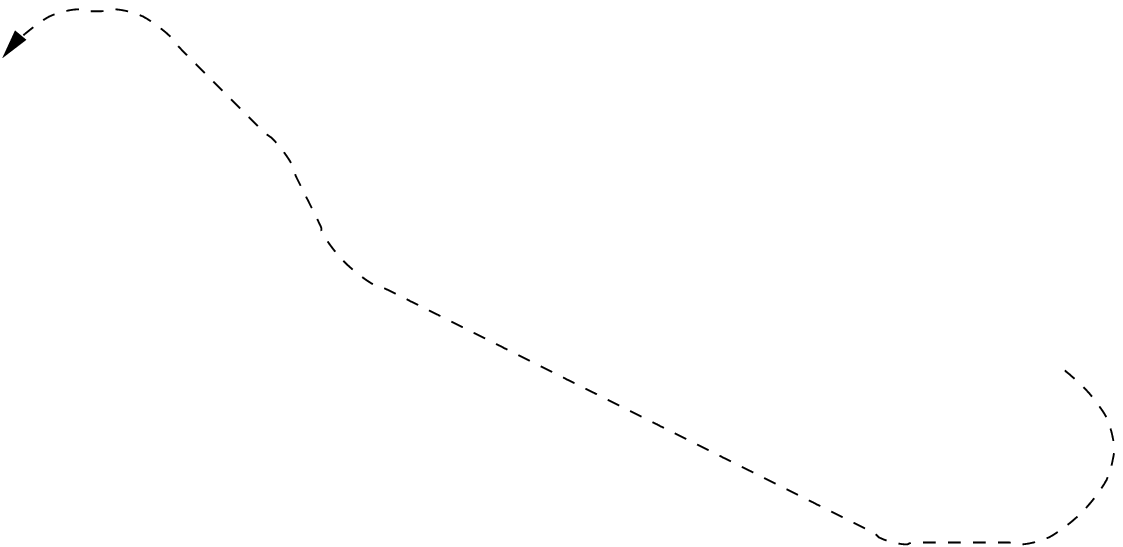}}}
\end{picture}
\end{center}

\begin{center}
\begin{picture}(6.5,2.25)(-0.25,-0.75)
\thicklines
\put(0.4,1){\makebox(0,0)[r]{(21b):}}
\put(1.125,1.0){\makebox(0,0){$I(B;D)$}}
\put(0.75,0.625){\line(1,2){0.125}}
\put(0.5,0.5){\makebox(0,0){$I(A;\Pair{B}{C}|B)$}}
\put(1.5,0.625){\line(-1,2){0.125}}
\put(1.75,0.5){\makebox(0,0){$I(A;\Pair{C}{D}|D)$}}
\put(2.125,0.125){\line(-1,2){0.125}}
\put(2.375,0){\makebox(0,0){$I(\Pair{B}{C};\Pair{C}{E}|\Pair{C}{D})$}}
\put(2.75,-0.375){\line(-1,2){0.125}}
\put(3,-0.5){\makebox(0,0){$I(A;\Pair{B}{C}|\Pair{C}{E})$}}
\put(4.125,1){\makebox(0,0){$I(C;E)$}}
\put(3.75,0.625){\line(1,2){0.125}}
\put(4.5,0.625){\line(-1,2){0.125}}
\put(3.5,0.5){\makebox(0,0){$I(A;\Pair{B}{C}|C)$}}
\put(5.125,0.125){\line(-1,2){0.125}}
\put(5.375,0){\makebox(0,0){$I(\Pair{B}{D};\Pair{C}{E}|\Pair{D}{E})$}}
\put(4.75,0.5){\makebox(0,0){$I(A;\Pair{D}{E}|E)$}}
\put(3.25,0.5){\makebox(0,0){\epsfxsize=3.8\unitlength\epsffile{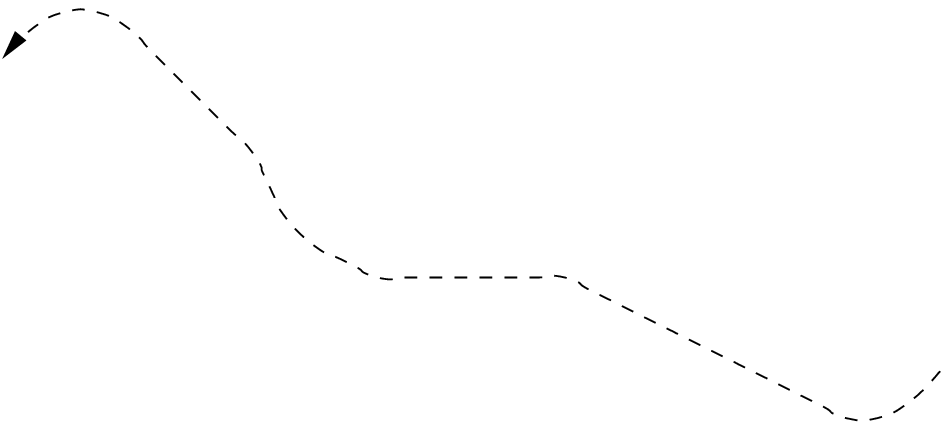}}}
\put(5.3125,0.5){\makebox(0,0){\epsfxsize=1.9\unitlength\epsffile{figs/pointer06.eps}}}
\end{picture}
\end{center}

\begin{center}
\begin{picture}(6.5,3)(-0.25,-0.5)
\thicklines
\put(0.32,2){\makebox(0,0)[r]{(22b):}}
\put(1.0,1.0){\makebox(0,0){$I(B;E)$}}
\put(0.625,0.625){\line(1,2){0.125}}
\put(0.375,0.5){\makebox(0,0){$I(A;\Pair{B}{C}|B)$}}
\put(1.375,0.625){\line(-1,2){0.125}}
\put(1.625,0.5){\makebox(0,0){$I(A;\Pair{D}{E}|E)$}}
\put(2,0.125){\line(-1,2){0.125}}
\put(2.25,0){\makebox(0,0){$I(\Pair{B}{E};\Pair{C}{D}|\Pair{D}{E})$}}
\put(2.75,1.5){\makebox(0,0){$I(C;D)$}}
\put(2.375,1.125){\line(1,2){0.125}}
\put(2.125,1){\makebox(0,0){$I(A;\Pair{B}{C}|C)$}}
\put(3.125,1.125){\line(-1,2){0.125}}
\put(3.375,1){\makebox(0,0){$I(A;\Pair{B}{D}|D)$}}
\put(3.75,0.625){\line(-1,2){0.125}}
\put(4,0.5){\makebox(0,0){$I(\Pair{B}{C};\Pair{B}{E}|\Pair{B}{D})$}}
\put(4.5,2.0){\makebox(0,0){$I(B;E)$}}
\put(4.125,1.625){\line(1,2){0.125}}
\put(3.875,1.5){\makebox(0,0){$I(A;\Pair{B}{C}|B)$}}
\put(4.875,1.625){\line(-1,2){0.125}}
\put(5.125,1.5){\makebox(0,0){$I(A;C|E)$}}
\put(5.5,1.125){\line(-1,2){0.125}}
\put(5.75,1){\makebox(0,0){$I(A;\Pair{B}{C}|C)$}}
\put(0.875,0.5){\makebox(0,0){\epsfxsize=2.3\unitlength\epsffile{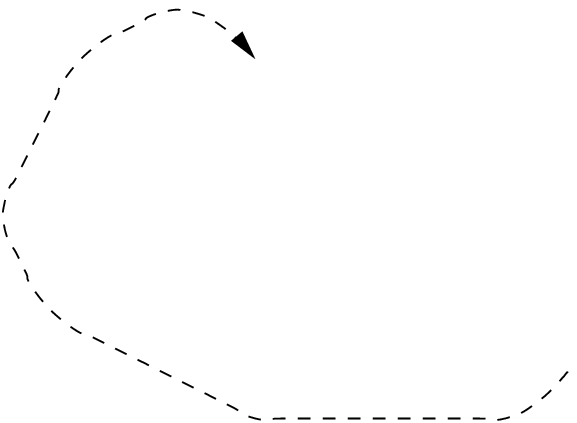}}}
\put(2.3125,0.75){\makebox(0,0){\epsfxsize=1.8\unitlength\epsffile{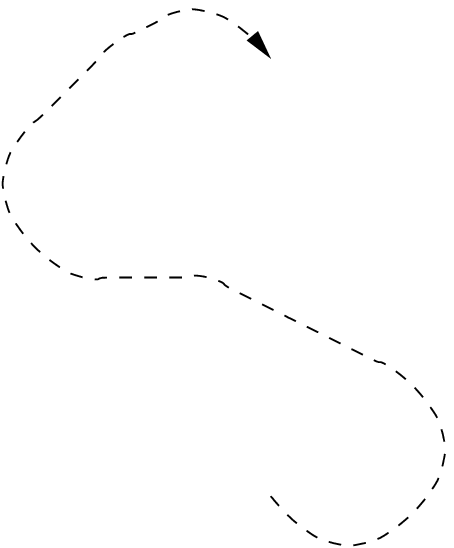}}}
\put(4.0625,1.25){\makebox(0,0){\epsfxsize=1.8\unitlength\epsffile{figs/pointer22b.eps}}}
\end{picture}
\end{center}

\begin{center}
\begin{picture}(6.5,2.5)(-0.125,-1)
\thicklines
\put(0.525,1){\makebox(0,0)[r]{(23b):}}
\put(1.125,1.0){\makebox(0,0){$I(B;E)$}}
\put(0.75,0.625){\line(1,2){0.125}}
\put(0.5,0.5){\makebox(0,0){$I(A;\Pair{B}{C}|B)$}}
\put(1.5,0.625){\line(-1,2){0.125}}
\put(1.75,0.5){\makebox(0,0){$I(A;\Pair{C}{E}|E)$}}
\put(2.125,0.125){\line(-1,2){0.125}}
\put(2.375,0){\makebox(0,0){$I(\Pair{B}{C};\Pair{C}{D}|\Pair{C}{E})$}}
\put(2.75,-0.375){\line(-1,2){0.125}}
\put(3,-0.5){\makebox(0,0){$I(A;\Pair{B}{C}|\Pair{C}{D})$}}
\put(4.125,1){\makebox(0,0){$I(D;E)$}}
\put(3.625,0.625){\line(1,1){0.25}}
\put(4.625,0.625){\line(-1,1){0.25}}
\put(3.375,0.5){\makebox(0,0){$I(A;\Pair{B}{D}|D)$}}
\put(3.75,0.125){\line(-1,2){0.125}}
\put(4,0){\makebox(0,0){$I(\Pair{B}{C};\Pair{B}{E}|\Pair{B}{D})$}}
\put(5.25,0.125){\line(-1,2){0.125}}
\put(5.5,0){\makebox(0,0){$I(\Pair{B}{C};\Pair{D}{E}|\Pair{C}{E})$}}
\put(4.875,0.5){\makebox(0,0){$I(A;\Pair{C}{E}|E)$}}
\put(2.92,0.25){\makebox(0,0){\epsfxsize=3.1\unitlength\epsffile{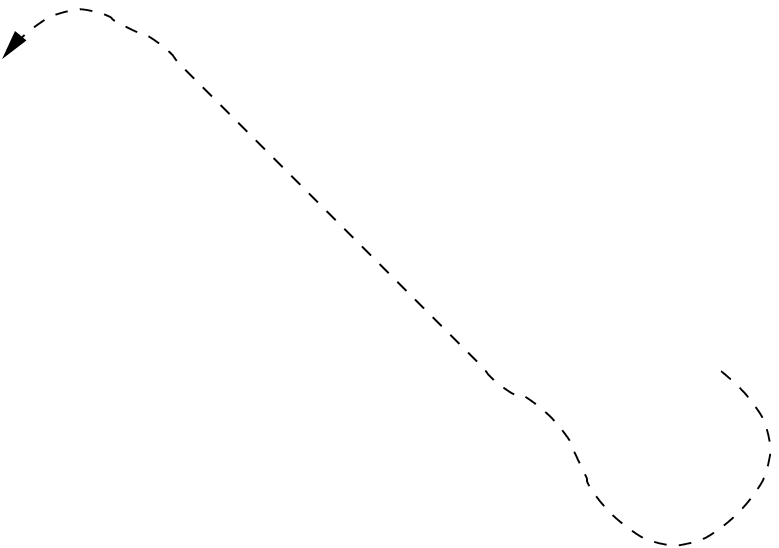}}}
\put(5.375,0.5){\makebox(0,0){\epsfxsize=2.1\unitlength\epsffile{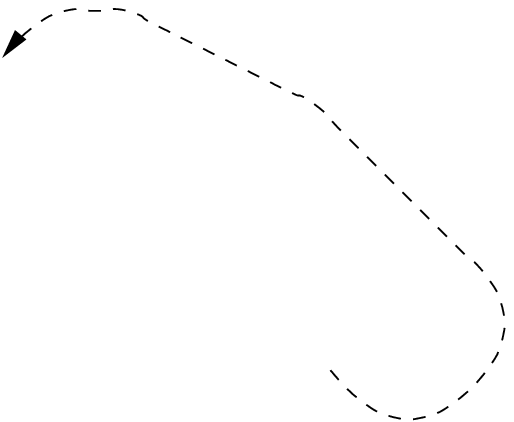}}}
\end{picture}
\end{center}

\begin{center}
\begin{picture}(6.5,2.25)(-0.25,-0.75)
\thicklines
\put(0.4,1){\makebox(0,0)[r]{(24b):}}
\put(1.125,1.0){\makebox(0,0){$I(A;E)$}}
\put(0.75,0.625){\line(1,2){0.125}}
\put(0.5,0.5){\makebox(0,0){$I(\Pair{A}{B};\Pair{C}{D}|A)$}}
\put(1.5,0.625){\line(-1,2){0.125}}
\put(1.75,0.5){\makebox(0,0){$I(\Pair{B}{E};C|E)$}}
\put(2.25,0.125){\line(-1,1){0.25}}
\put(2.5,0){\makebox(0,0){$I(\Pair{A}{B};\Pair{C}{D}|C)$}}
\put(1.25,0.125){\line(1,1){0.25}}
\put(1,0){\makebox(0,0){$I(\Pair{A}{B};\Pair{B}{D}|\Pair{B}{E})$}}
\put(1.375,-0.375){\line(-1,2){0.125}}
\put(1.625,-0.5){\makebox(0,0){$I(\Pair{A}{B};\Pair{C}{D}|\Pair{B}{D})$}}
\put(4.125,1){\makebox(0,0){$I(B;D)$}}
\put(3.625,0.625){\line(1,1){0.25}}
\put(4.5,0.625){\line(-1,2){0.125}}
\put(4.75,0.5){\makebox(0,0){$I(\Pair{C}{D};\Pair{D}{E}|D)$}}
\put(5.125,0.125){\line(-1,2){0.125}}
\put(5.375,0){\makebox(0,0){$I(\Pair{A}{E};\Pair{B}{D}|\Pair{D}{E})$}}
\put(3.375,0.5){\makebox(0,0){$I(\Pair{A}{B};\Pair{C}{D}|B)$}}
\put(3.25,0.5){\makebox(0,0){\epsfxsize=3.8\unitlength\epsffile{figs/pointer21.eps}}}
\put(5.3125,0.5){\makebox(0,0){\epsfxsize=1.9\unitlength\epsffile{figs/pointer06.eps}}}
\end{picture}
\end{center}

One can also get a new extended version of Theorem~\ref{thm:list1}
in the same way, though it is harder to state precisely.  It is also
slightly less flexible because it disallows reuse of the same variable
or combination of variables; and the forest diagrams are easier to
verify by inspection.

Here are two more explicit proofs.

\begin{proof}[Proof of inequality (18)]
Let $Z$ be a common information of $A$ and $B$, and let $Y$ be
a common information of $A$ and $C$; note that we have
$H(\Pair{Y}{Z}|A) = 0$.  Then
\begin{align*}
&         I(B;C)+I(A;B|D)+I(A;C|D)+I(B;D|E) \notag\\
         &\Gap +I(C;D|E)+I(A;E)\\
      &\ge I(Z;Y)+I(\Pair{Y}{Z};Z|D)+I(\Pair{Y}{Z};Y|D)+I(Z;D|E)\notag\\
            &\Gap +I(Y;D|E)+I(\Pair{Y}{Z};E)   &\Comment{Corollary \ref{cor:1}}\\
       &= I(Z;Y)+H(Z|D)+H(Y|D)+I(Z;D|E) \notag\\
           &\Gap +I(Y;D|E)+I(\Pair{Y}{Z};E)\\
      &\ge I(Z;Y)+H(Y|D)+I(Z;Z|E)+I(Y;D|E)+I(\Pair{Y}{Z};E)    &\Comment{Lemma \ref{lem:RtoZ}}\\
      &\ge I(Z;Y)+I(Z;Z|E)+I(Y;Y|E)+I(\Pair{Y}{Z};E)           &\Comment{Lemma \ref{lem:RtoZ}}\\
       &= I(Z;Y)+H(Z|E)+H(Y|E)+I(\Pair{Y}{Z};E)\\
      &\ge I(Z;Y)+H(\Pair{Y}{Z}|E)+I(\Pair{Y}{Z};E)\\
       &= I(Z;Y)+H(\Pair{Y}{Z})\\
       &= H(Z)+H(Y)\\
       &= I(A;B)+I(A;C).
\end{align*}
\end{proof}

\begin{proof}[Proof of inequality (20)]
Let $Z$ be a common information of $A$ and $B$, and let $Y$ be
a common information of $A$ and $C$; note that we have
$H(\Pair{Y}{Z}|A) = 0$ and $H(\Pair{C}{Y}|C) = H(C|\Pair{C}{Y}) = 0$.  Then
\begin{align*}
&         I(B;C)+I(B;D)+I(A;C|D) \notag\\
     &\Gap +I(A;B|E)+I(A;E|B)+ I(C;D|E)+I(B;E|\Pair{C}{D})\\
      &\ge I(B;Y)+I(Z;D)+I(\Pair{Y}{Z};\Pair{C}{Y}|D)\notag\\
            &\Gap +I(Z;Z|E)+I(Y;E|B)+ I(Y;D|E)+I(Z;E|\Pair{C}{D})   &\Comment{Corollary \ref{cor:1}}\\
      &= I(B;Y)+I(Z;D)+I(\Pair{Y}{Z};\Pair{C}{Y}|D)\notag\\
            &\Gap +I(Z;Z|E)+I(Y;E|B)+ I(Y;D|E)+I(Z;E|\Triple{C}{Y}{D}) \notag\\
      &= I(\Pair{B}{E};Y)+I(Z;D)+I(\Pair{Y}{Z};\Pair{C}{Y}|D)\notag\\
            &\Gap +I(Z;Z|E)+ I(Y;D|E)+I(Z;E|\Triple{C}{Y}{D}) \notag\\
      &\ge I(E;Y)+I(Z;D)+I(\Pair{Y}{Z};\Pair{C}{Y}|D)\notag\\
            &\Gap +I(Z;Z|E)+ I(Y;D|E)+I(Z;E|\Triple{C}{Y}{D}) \notag\\
      &= I(\Pair{D}{E};Y)+I(Z;D)+I(\Pair{Y}{Z};\Pair{C}{Y}|D)\notag\\
            &\Gap +I(Z;Z|E)+I(Z;E|\Triple{C}{Y}{D}) \notag\\
      &= I(\Pair{D}{E};Y)+I(Z;D)+I(Z;\Pair{C}{Y}|D)+I(Y;\Pair{C}{Y}|\Pair{D}{Z})\notag\\
            &\Gap +I(Z;Z|E)+I(Z;E|\Triple{C}{Y}{D}) \notag\\
      &= I(\Pair{D}{E};Y)+I(Z;D)+I(Z;\Pair{C}{Y}|D)+H(Y|\Pair{D}{Z})\notag\\
            &\Gap +H(Z|E)+I(Z;E|\Triple{C}{Y}{D}) \notag\\
      &= I(\Pair{D}{E};Y)+I(Z;D)+I(Z;\Triple{C}{E}{Y}|D)+H(Y|\Pair{D}{Z})
            +H(Z|E) \notag\\
      &\ge I(\Pair{D}{E};Y)+I(Z;D)+I(Z;\Pair{E}{Y}|D)+H(Y|\Pair{D}{Z})
            +H(Z|E) \notag\\
      &= I(\Pair{D}{E};Y)+I(Z;\Triple{D}{E}{Y})+H(Y|\Pair{D}{Z})
            +H(Z|E) \notag\\
      &= I(\Pair{D}{E};Y)+I(Z;\Pair{D}{E})+I(Z;Y|\Pair{D}{E})+H(Y|\Pair{D}{Z})
            +H(Z|E) \notag\\
      &\ge I(\Pair{D}{E};Y)+I(Z;\Pair{D}{E})+I(Z;Y|\Pair{D}{E})\notag\\
            &\Gap +H(Y|\Triple{D}{E}{Z})+H(Z|\Pair{D}{E}) \notag\\
      &= I(\Pair{D}{E};Y)+I(Z;\Pair{D}{E})+H(Y|\Pair{D}{E})
            +H(Z|\Pair{D}{E}) \notag\\
      &= I(\Pair{D}{E};Y)+H(Z)+H(Y|\Pair{D}{E}) \notag\\
       &= H(Z)+H(Y)\\
       &= I(A;B)+I(A;C).
\end{align*}
\end{proof}

It is not yet clear how to
generalize these.

\section{Completeness}

The complete (and verified nonredundant) list of linear-variable
inequalities on five variables consists of:
\begin{itemize}
\item the elemental Shannon inequalities:
\begin{align*}
       0 &\le I(A;B)\\
       0 &\le I(A;B|C)\\
       0 &\le I(A;B|\Pair{C}{D})\\
       0 &\le I(A;B|\Triple{C}{D}{E})\\
       0 &\le H(A|\Quadruple{B}{C}{D}{E})
\end{align*}
and the inequalities obtained from these by permuting the five variables
$A,B,C,D,E$ (see Yeung~\cite{Yeung-book}
for a proof that these imply all other 5-variable Shannon
inequalities);
\item the following instances of the Ingleton inequality:
\begin{align}
\label{eqn:inginst1}
       I(A;B) &\le I(A;B|C) + I(A;B|D) + I(C;D)\\
       I(A;B) &\le I(A;B|C) + I(A;B|\Pair{D}{E}) + I(C;\Pair{D}{E})\\
       I(A;\Pair{B}{C}) &\le I(A;\Pair{B}{C}|D) + I(A;\Pair{B}{C}|E) + I(D;E)\\
\label{eqn:inginst4}
       I(\Pair{A}{B};\Pair{A}{C}) &\le I(\Pair{A}{B};\Pair{A}{C}|\Pair{A}{D}) + I(\Pair{A}{B};\Pair{A}{C}|\Pair{A}{E}) + I(\Pair{A}{D};\Pair{A}{E})
\end{align}
and the ones obtained from these by permuting the five variables
$A,B,C,D,E$ (see Guill\'e, Chan, and Grant~\cite{Guille-Chan-Grant}
for a proof that these imply all other 5-variable instances of the Ingleton
inequality); and
\item inequalities (1)--(24) and their
permuted-variable forms. \end{itemize}

To verify the completeness of this list, we consider the
31-dimensional real space whose coordinates are labeled by the
subsets of $\{A,B,C,D,E\}$ in the usual binary order:
$$\{A\}, \{B\}, \{A,B\}, \{C\}, \{A,C\}, \{B,C\}, \ldots, \{A,B,C,D,E\}.$$
Each of the listed inequalities, once it is
rewritten in terms of the basic entropy terms
\begin{equation}
\label{eq:entlist}
H(A), H(B),
H(\Pair{A}{B}),H(C),H(\Pair{A}{C}),\ldots,
H(\Triple{\Triple{A}{B}{C}}{D}{E}),
\end{equation}
defines a half-space of this
space; the intersection of these half-spaces is a polyhedral cone
which can also be described as the convex hull of its extreme rays.
If one of these extreme rays contains a nonzero point $v$ which is
\textit{(linearly) representable} (i.e., there exist a vector space
$U$ and subspaces $U_A,U_B,U_C,U_D,U_E$ of $U$ such that
$\dim(U_A)=v(A)$, $\dim(U_B)=v(B)$, $\dim(\langle
U_A,U_B\rangle)=v(\Pair{A}{B})$, and so on), then this extreme ray
can never be excluded by any as-yet-unknown linear rank inequality.
If we verify that \textit{all} of the extreme rays contain linearly
representable points, then there can be no linear rank inequality
which cuts down the polyhedral cone further, so the list of
inequalities must be complete.

There are 7943 extreme rays in $\R^{31}$ determined by the elemental
Shannon inequalities and inequalities (1)--(24) and
\eqref{eqn:inginst1}--\eqref{eqn:inginst4} (and permutations).
If one considers two such rays to be essentially the same when one can be
obtained from the other by a permutation of the five variables, then
there are 162 essentially different extreme rays.  A full list of
the vectors generating these rays
is available at:
\begin{center}
\filewebsite
\end{center}

The authors have shown that each of these vectors is representable over
the field of real numbers; in fact, up to a scalar multiple,
this representation can be done using
matrices with integer entries which actually represent the vector
over any field (finite or infinite).  For instance, consider the extreme ray
given by the vector
\begin{align*}
   1\ 1\ 2\ 1\ 2\ 2\ 3\ 1\ 2\ 2\ 3\ 2\ 3\ 3\ 3
   \ 2\ 3\ 3\ 3\ 2\ 3\ 3\ 3\ 2\ 3\ 3\ 3\ 2\ 3\ 3\ 3
\end{align*}
(a list of 31 ranks or entropies in the order given by \eqref{eq:entlist}).
To this we associate the five matrices:
\begin{align*}
    M_A &=
\left[
 \begin{array}{rrr}
  1 & 0 & 0\\
 \end{array}
\right]\\
    M_B &=
\left[
 \begin{array}{rrr}
  0 & 1 & 0\\
 \end{array}
\right]\\
    M_C &=
\left[
 \begin{array}{rrr}
  0 & 0 & 1\\
 \end{array}
\right]\\
    M_D &=
\left[
 \begin{array}{rrr}
  1 & 1 & 1\\
 \end{array}
\right]\\
    M_E &=
\left[
 \begin{array}{rrr}
  1 & 1 & 0\\
  0 & 0 & 1\\
 \end{array}
\right]
\end{align*}
The interpretation here is that we have a fixed field $F$, and the
row space of each of these matrices specifies a subspace of $F^3$.
The specified vector gives $H(A)=1$, and the row space of $M_A$ has
dimension 1; the vector gives $H(B)=1$, and the row space of $M_B$
has dimension 1; the vector gives $H(\Pair{A}{B})=2$, and the vector sum of
the row spaces of $M_A$ and $M_B$ (i.e., the row space of
$M_A$-on-top-of-$M_B$) has dimension 2; and so on.  Equivalently, if
we take three random variables $x_1,x_2,x_3$ chosen uniformly and
independently over the finite field $F$, and let $A = x_1$, $B = x_2$,
$C = x_3$, $D = x_1+x_2+x_3$, and $E = (x_1+x_2,x_3)$, then the
entropies of all combinations of $A,B,C,D,E$ (with logarithms to
base $|F|$) are as specified by the above vector.

The dimensions of the row spaces listed above are easily computed over
the real field (as ranks of the corresponding matrices).  In order to
verify that the same dimensions would be obtained over any field, one just
has to note that, in each case where a matrix rank is computed to be $k$,
there is actually a $k\times k$ submatrix whose determinant is $\pm1$, so the
selected $k$ rows will still be independent even after being reduced modulo
any prime.  (Actually, it would suffice to verify that the greatest common
divisor of the determinants of all $k\times k$ submatrices is 1.)

All of the other listed vectors turn out to be representable in the same way,
except that for a few of them a scalar multiplier must be applied.  For
instance, consider the vector
\begin{align*}
   0\ 1\ 1\ 1\ 1\ 2\ 2\ 1\ 1\ 2\ 2\ 2\ 2\ 2\ 2
   \ 1\ 1\ 2\ 2\ 2\ 2\ 2\ 2\ 2\ 2\ 2\ 2\ 2\ 2\ 2\ 2.
\end{align*}
To represent this, we would normally take $M_A$ to be a $0\times2$ matrix and
$M_B,M_C,M_D,M_E$ to be $1\times2$ matrices whose unique rows have the
property that any two are independent but any three are dependent.  (In other
words, these row vectors are a linear representation for the uniform matroid
$U_{2,4}$.)  For example, we could take
\begin{align*}
    M_A &= [\ ]\\
    M_B &=
\left[
 \begin{array}{rr}
  1 & 0\\
 \end{array}
\right]\\
    M_C &=
\left[
 \begin{array}{rr}
  0 & 1\\
 \end{array}
\right]\\
    M_D &=
\left[
 \begin{array}{rr}
  1 & 1\\
 \end{array}
\right]\\
    M_E &=
\left[
 \begin{array}{rr}
  1 & 2\\
 \end{array}
\right]
\end{align*}
over the real field, but these would not work over the field of two elements.
In fact, no such choice of row vectors works over the field of two elements
(the first two row vectors would be independent, but then the only choice for
the third vector would be the sum of the first two, and the same would hold
for the fourth vector, contradicting the independence of the third and
fourth vectors).  But if we instead take the vector
\begin{align*}
   0\ 2\ 2\ 2\ 2\ 4\ 4\ 2\ 2\ 4\ 4\ 4\ 4\ 4\ 4
   \ 2\ 2\ 4\ 4\ 4\ 4\ 4\ 4\ 4\ 4\ 4\ 4\ 4\ 4\ 4\ 4,
\end{align*}
which is twice the preceding vector and hence determines the same extreme ray,
then we can get suitable representing matrices
\begin{align*}
    M_A &= [\ ]\\
    M_B &=
\left[
 \begin{array}{rrrr}
  1 & 0 & 0 & 0\\
  0 & 1 & 0 & 0\\
 \end{array}
\right]\\
    M_C &=
\left[
 \begin{array}{rrrr}
  0 & 0 & 1 & 0\\
  0 & 0 & 0 & 1\\
 \end{array}
\right]\\
    M_D &=
\left[
 \begin{array}{rrrr}
  1 & 0 & 1 & 0\\
  0 & 1 & 0 & 1\\
 \end{array}
\right]\\
    M_E &=
\left[
 \begin{array}{rrrr}
  1 & 1 & 0 & 1\\
  0 & 1 & 1 & 0\\
 \end{array}
\right]
\end{align*}
which work over any field.  The same doubling is needed for 13 more of the
162 vectors; and one additional vector, the vector
\begin{align*}
   1\ 1\ 2\ 1\ 2\ 2\ 2\ 1\ 2\ 2\ 2\ 2\ 2\ 2\ 2
   \ 1\ 2\ 2\ 2\ 2\ 2\ 2\ 2\ 2\ 2\ 2\ 2\ 2\ 2\ 2\ 2
\end{align*}
corresponding to the uniform matroid $U_{2,5}$, had to be tripled in order
to get a matrix representation that works over all fields.

\section{Methodology; testing representability of polymatroids}

The list of five-variable linear rank inequalities was produced by the
following iterative process.  Initially, we had the Shannon and Ingleton
inequalities.  At each stage, we took the current list of inequalities
and used Komei Fukuda's {\tt cddlib} software~\cite{cddlib}
to get the corresponding list of extreme rays.  We then examined the
vectors generating the extreme rays to see whether they were representable
(over the reals;
we did not try to get representations working over all fields until after
the iterative process was complete).  When such a vector provably could not
be represented, the proof (in each case we ran into here) yielded a new
linear rank inequality provable via common informations; when we examined a
vector where we had difficulty determining whether it was representable or not,
we ran exhaustive tests on all ways of specifying a single common information
(toward the end, we had to try a pair of common informations) to see whether
{\tt ITIP} could verify that the specified vector contradicted the Shannon
inequalities together with the common information specification.  Again each
such verification led to a new linear rank inequality.  (Of course, this is
a highly sanitized version of the process as it actually occurred.)

The testing of extreme rays for linear representability soon became a large
task, so we gradually developed software to automatically find such
representations in a number of cases (and we added more cases when we found
new ways to represent vectors).  This software used combinatorial rather than
linear-algebra methods; for instance, the output of the program for the
sample vector
\begin{align}
\label{eq:PolymatExample}
   1\ 1\ 2\ 1\ 2\ 2\ 3\ 1\ 2\ 2\ 3\ 2\ 3\ 3\ 3
   \ 2\ 3\ 3\ 3\ 2\ 3\ 3\ 3\ 2\ 3\ 3\ 3\ 2\ 3\ 3\ 3
\end{align}
used above was a specification of five vector spaces $A,B,C,D,E$ which could
be paraphrased as: ``$A$ is generated by one vector, $B$ is generated by
one vector not in $A$, $C$ is generated by one vector not in $A+B$
[the space spanned by $A$ and $B$], $D$ is generated by one vector in general
position in $A+B+C$, and $E$ is generated by two vectors, one in
$(A+B)\cap(C+D)$ and one in $C$.''  The development of the
software involved recognizing as many cases as possible where one could find
such a specification which could be met over the reals (or over
any sufficiently large finite field) and would yield the
desired rank vector.

The (attempted) construction of a representation is done one basic subspace
at a time: first the representation of $A$ is constructed (this step is
trivial), then the representation of $B$ given $A$, then the representation
of $C$ given $A$ and $B$, and so on.  And each of these subspace representations
is constructed one basis vector at a time.  Given the representation of
$A$, $B$, $C$, and $D$, the algorithm will determine how many basis
vectors are needed for subspace $E$ and successively try to choose them in
suitable positions relative to the existing subspaces.  At each step,
a new vector will be chosen in general position in a subspace which is
a sum of some of the already-handled subspaces $A,B,C,D$.  (Here ``general
position'' means in the selected subspace but not in any relevant
proper subspace of it.  Which subspaces are relevant depends on the
current situation; we avoid having to determine this explicitly by
just saying that the underlying field is sufficiently large, or infinite.)
If there is a problem with specifying that the vector is in such a sum
of basic subspaces, then we may have to specify that the vector is in
the intersection of two sums of basic subspaces.

Once the first vector is chosen, we take quotients of all of the existing
spaces by this vector to get the new situation in which the second vector
needs to be chosen.  This is all done by counting dimensions, not by
constructing actual numerical vectors.  For instance, suppose the first
vector is chosen to be in general position in subspace $R$ which is
a sum of basic subspaces from $A,B,C,D$ (e.g., $R = A+B$).
For each other sum subspace $T$, if the new vector is in $T$, then
the quotient by the chosen vector will reduce the dimension of $T$ by 1;
if the chosen vector is not in $T$, then the quotient will not change the
dimension of $T$.  Since the vector is in general position in $R$,
the vector will be in $T$ if and only if $R \subseteq T$, and to check
whether $R \subseteq T$ one simply has to see whether $\dim(R+T) = \dim T$.
The case where the vector is chosen from an intersection of two
sum subspaces $R$ and $S$ is more complicated; more on this below.

Consider the example \eqref{eq:PolymatExample}.  Suppose that we
have already constructed the representations for subspaces $A$, $B$,
$C$, and $D$, and we are now ready to construct the representation
for subspace~$E$.  The current situation can be summarized by the
following two-row array:
\begin{align}
\label{eq:MSarray1}
\begin{array}{rrrrrrrrrrrrrrrr}
0 & 1 & 1 & 2 & 1 & 2 & 2 & 3 & 1 & 2 & 2 & 3 & 2 & 3 & 3 & 3\\
2 & 2 & 2 & 1 & 1 & 1 & 1 & 0 & 1 & 1 & 1 & 0 & 0 & 0 & 0 & 0
\end{array}
\end{align}
Here the first row is the ranks of sums from $A,B,C,D$ in the
order given by \eqref{eq:entlist}, but starting with the
empty space.  For each of these sums, the second row gives
the amount by which adding the new subspace $E$ will increase
the dimension of the sum.  (So the second entry in this row
is $H(A+E)-H(A) = 3-1 = 2$, the fourth entry is
$H(A+B+E)-H(A+B) = 3 - 2 = 1$, and so on.)

From this array, we can see that, since $E$ has dimension 2 but only increases
the dimension of $A+B$ by 1, one of the nonzero vectors in $E$ must be
in $A+B$.  So let us start by assuming that one of the vectors in $E$
is a vector chosen in general position in $R = A+B$.  We can now check
for all sums from $A,B,C,D$ whether the sum will contain this chosen
vector; this information is summarized in the row
\begin{align}
\label{eq:diffvec1}
\begin{array}{rrrrrrrrrrrrrrrr}
0 & 0 & 0 & 1 & 0 & 0 & 0 & 1 & 0 & 0 & 0 & 1 & 0 & 1 & 1 & 1
\end{array}
\end{align}
where 1 means the chosen vector is in the corresponding sum.
To get the result of taking a quotient by (the subspace generated by)
the chosen vector, we subtract \eqref{eq:diffvec1} from the first row
of~\eqref{eq:MSarray1} (because we have used up one vector from each
of the indicated subspaces) and subtract the one's complement
of~\eqref{eq:diffvec1} from the second row
of~\eqref{eq:MSarray1} (because we have taken care of one of the new
vectors for $E$ beyond each of the indicated subspaces).  So the
situation after the first vector is chosen is given by:
\begin{align*}
\label{eq:MSarray2}
\begin{array}{rrrrrrrrrrrrrrrr}
0 & 1 & 1 & 1 & 1 & 2 & 2 & 2 & 1 & 2 & 2 & 2 & 2 & 2 & 2 & 2\\
1 & 1 & 1 & 1 & 0 & 0 & 0 & 0 & 0 & 0 & 0 & 0 & -1 & 0 & 0 & 0
\end{array}
\end{align*}

Of course, the negative entry in this array means that a problem has
occurred: we tried to take a new vector not in $C+D$, but the
given ranks require all vectors in $E$ to be in $C+D$.  So we will
try again; instead of taking a vector in general position in $R=A+B$,
we take a vector in general position in $R \cap S$, where $S = C+D$.

This leaves the problem of determining, for each sum subspace $T$,
whether the chosen vector is in $T$; as before, this is equivalent
to determining whether $R\cap S \subseteq T$.  This is not as
straightforward as it was to determine whether $R\subseteq T$;
in fact, there are situations where the given data on ranks of
sum subspaces simply do not determine whether $R \cap S \subseteq T$.
But we have identified many situations where the given data
do allow this determination to be made.  Here is a list; note that
(a) each such test can also be applied with $R$ and $S$ interchanged, and
(b) reading this list is not necessary for understanding the rest
of the algorithm.

\begin{itemize}
\item If $R \subseteq T$, then $R \cap S \subseteq T$.
\item If the dimensions of $R \cap S$, $R \cap T$, and $R \cap (S+T)$ are
all equal, then $R \cap S \subseteq T$.  [If two subspaces have the same
(finite) dimension and one is included in the other, then the two subspaces
are equal.  Hence, we get
$R \cap S = R \cap (S+T) = R \cap T$, so $R \cap S = (R \cap S) \cap (R \cap T)
= R \cap S \cap T$,
so $R \cap S \subseteq T$.  Also, recall that the dimension of $R\cap S$
can be determined from the given data; it is equal to
$I(R;S) = H(R)+H(S)-H(\Pair{R}{S})$.]
\item If the dimensions of $R\cap T$, $S\cap T$, $(R+S)\cap T$,
and $R\cap S$ are all equal, then $R \cap S \subseteq T$. [We have
$R\cap T = (R+S)\cap T = S\cap T$, so $R\cap T = R\cap S\cap T$.
But now $\dim(R\cap S)=\dim(R\cap T)=\dim(R\cap S\cap T)$, so
$R\cap S=R\cap S\cap T$, so $R \cap S \subseteq T$.]
\item If $\dim(R \cap T) < \dim(R \cap S)$, then $R \cap S \not\subseteq R \cap T$,
so we must have $R \cap S \not\subseteq T$.
\item Let $R \cap^* S$ be the ``nominal intersection'' of $R$ and $S$ (i.e., the
sum of the basic subspaces listed both in the sum $R$ and the sum $S$).
Clearly $R \cap^* S \subseteq R \cap S$, so, if $R \cap^* S \not\subseteq T$,
then $R \cap S \not\subseteq T$.
\item If $\dim(R \cap T) < \dim(R \cap ((R\cap^*T)+S))$,
then $R \cap S \not\subseteq T$.  [First note that, if $U,V,W$ are subspaces
such that $V \subseteq U$, then $U \cap (V+W) = V+(U\cap W)$.  (The right-to-left
inclusion is easy.  For the left-to-right inclusion, if $u=v+w$ where $u\in U$,
$v\in V$, and $w \in W$, then $u-v = w \in U \cap W$, so $v+w \in V+(U\cap W)$.)
Hence, if $R \cap S \subseteq T$, then $R \cap ((R \cap^* T)+S)
= (R \cap^* T)+(R \cap S) \subseteq R \cap T$, so
$\dim(R \cap ((R\cap^*T)+S)) \le \dim(R \cap T)$.]
\item
If $T' \subseteq T$ and $R \cap S \subseteq T'$,
then $R \cap S \subseteq T$.
If $T \subseteq T'$ and $R \cap S \not\subseteq T'$,
then $R \cap S \not\subseteq T$.
\item Let $R \setminus^* S$ be the ``nominal difference'' of $R$ and $S$ (i.e., the
sum of the basic subspaces listed in the sum $R$ but not in the sum $S$),
and let $U = (R \setminus^* S)+(S \setminus^* R)$.  If
$\dim(U \cap (R \cap^* S)) = 0$, then
$$R \cap S = ((R \setminus^* S)\cap(S \setminus^* R)) + (R \cap^* S).$$
[The right-to-left
inclusion is easy.  For the left-to-right inclusion,
note that $R = (R \setminus^* S)+(R \cap^* S)$
and $S = (S \setminus^* R)+(R \cap^* S)$.
Hence, if $x \in R \cap S$,
then we we have $x=y_1+z_1=y_2+z_2$ for some $y_1 \in R \setminus^* S$,
$y_2 \in S \setminus^* R$, and
$z_1,z_2 \in R \cap^* S$. Then $y_2-y_1 = z_1-z_2$ is in
$U \cap (R \cap^* S)$, so we have $y_2=y_1$ and $z_2=z_1$; hence,
$y_1 \in(R \setminus^* S)\cap(S \setminus^* R)$ and $x=y_1+z_1$ is
in the desired form.]
Hence, if $\dim(U \cap (R \cap^* S)) = 0$, $R \cap^* S \subseteq T$,
and $((R \setminus^* S)\cap(S \setminus^* R)) \subseteq T$, then
$R \cap S \subseteq T$.
\end{itemize}

These tests do suffice for the example here; the resulting
membership vector is
\begin{align*}
\begin{array}{rrrrrrrrrrrrrrrr}
0 & 0 & 0 & 1 & 0 & 0 & 0 & 1 & 0 & 0 & 0 & 1 & 1 & 1 & 1 & 1
\end{array}
\end{align*}
and the new array after taking a quotient by the first chosen vector is:
\begin{align*}
\begin{array}{rrrrrrrrrrrrrrrr}
0 & 1 & 1 & 1 & 1 & 2 & 2 & 2 & 1 & 2 & 2 & 2 & 1 & 2 & 2 & 2\\
1 & 1 & 1 & 1 & 0 & 0 & 0 & 0 & 0 & 0 & 0 & 0 & 0 & 0 & 0 & 0
\end{array}
\end{align*}
Let us call the new quotient spaces $A',B',C',D',E'$.  The new
ranks indicate that the remaining vector in $E'$ must be chosen
to be in $C'$.  If we take the new vector in general position
in $C'$, then the resulting membership vector is:
\begin{align*}
\begin{array}{rrrrrrrrrrrrrrrr}
0 & 0 & 0 & 0 & 1 & 1 & 1 & 1 & 1 & 1 & 1 & 1 & 1 & 1 & 1 & 1
\end{array}
\end{align*}
(Note that we needed the chosen vector to be in $D'$ as well
as in $C'$, but this turned out to be automatic, because
the given ranks implied $C' = C'+D' = D'$.)  And the result
of taking a quotient by the second chosen vector is:
\begin{align*}
\begin{array}{rrrrrrrrrrrrrrrr}
0 & 1 & 1 & 1 & 0 & 1 & 1 & 1 & 0 & 1 & 1 & 1 & 0 & 1 & 1 & 1\\
0 & 0 & 0 & 0 & 0 & 0 & 0 & 0 & 0 & 0 & 0 & 0 & 0 & 0 & 0 & 0
\end{array}
\end{align*}
The all-0 row means that the representation of $E$ has been successfully
completed.

The current algorithm does not try many possibilities for the next vector
to choose; it simply chooses one sum subspace (usually at the beginning
of the list of available ones) to try to add a vector to, and, if that
yields an immediate contradiction, perhaps tries one intersection of
two sum subspaces.  If any such step fails (either because of a
contradiction, or because the algorithm cannot determine whether
$R \cap S \subseteq T$ in some case), the algorithm gives up.
However, the algorithm does give itself up to 120 chances by trying
all permutations of the 5 basic variables.

Each time a new extreme ray was produced, the above algorithm was
applied as a positive test for representability, while tests against
common informations were used as negative tests.  If both sides
failed, the ray was examined by hand.  Sometimes this examination
yielded a representation because we found a new way of
determining whether $R\cap S\subseteq T$; if so, this new test was added
to the algorithm.  At the end, the algorithm was able to verify
representability of 152 of the final 162 extreme rays, leaving
only 10 to be done by hand (by methods which did not fit in
the framework of this algorithm).

There are other possibilities for improving the algorithm that
we have not yet implemented.  One is doing a backtrack search to
consider more possibilities for choosing vectors to add; another is
to use the information on representation of previous subspaces
in the construction of the representation of the current subspace.
(In the preceding example, we used only the dimension data for
$A,B,C,D$ in the construction of the representation for $E$; we
did not use the actual representations constructed for $A,B,C,D$.)
More ambitious would be to allow more options for choosing
new vectors in terms of the known relations between the current
subspaces.

\section{Six-variable inequalities (ongoing work)}
\label{sec:6var}

This iterative process for finding all linear rank inequalities is likely to
be infeasible to complete for six or more variables.  (Each {\tt cddlib} polytope computation
in 31 dimensions took about 2--3 days; in 63 dimensions it would take
far longer, as well as rapidly exceeding the memory available.)
But we plan to continue the study, because we expect to find
new phenomena at higher levels, possibly including extreme rays that are
representable over some fields but not over others (hence yielding
rank inequalities which hold only over those other fields), and inequalities
which hold for ranks of vector spaces but are not provable via common
informations.  For instance, such situations could come from the variables
associated with the Fano and non-Fano networks
in \cite{Dougherty-Freiling-Zeger-matroidal}, or the network in
\cite{Dougherty-Freiling-Zeger04}.

In order to make any progress at all, we had to take some shortcuts
(since, as noted above, 63-dimensional polytope computations were out
of the question).  One of these was to reduce the dimension of the search
by assuming equality for one or more of the inequalities found so far;
in effect, this is just concentrating on one face, corner, or
intermediate-dimensional extreme part of the current region.
Another was to work hard on trying to improve already-obtained inequalities,
find additional instances of them, or strengthen them in multiple ways
if they were not already faces of the region.

We will show here some of the 6-variable inequalities we have found so far;
a much longer list is available at:
\begin{center}
\filewebsite
\end{center}
All of these have been verified
to be faces of the linear rank region (so they cannot be improved).  To do
this, we used a stockpile of linearly representable 6-variable polymatroids
(the representability was proved by the algorithm described in the
preceding section) encountered during the polytope computations.  If a
6-variable linear rank inequality is satisfied with equality by
62 linearly independent vectors from the stockpile, then it must give
a face of the linear rank region.  (The stockpile currently contains
3220 polymatroids, or 1846734 after
one takes all instances obtained by permuting the six basic variables.  It is
also available at the above website.)

First, there are the 6-variable elemental Shannon inequalities; there are
6 of these if one lists just one of each form, but 246 of them if all
of the permuted-variable versions are counted.  Then there are
12 instances of the Ingleton inequality (1470 counting permuted forms).
Again, see Yeung~\cite{Yeung-book} and Guill\'e, Chan, and
Grant~\cite{Guille-Chan-Grant} for the proof that these inequalities
imply all of the other Shannon and Ingleton inequalities.

Next come the instances of the 5-variable inequalities (1)--(24).
The initial computation found 183 of these instances that (with permuted
forms)
proved all of the others.  However, 16 of these instances did not
pass the face verification above and were later superseded by
other 6-variable inequalities; this left 167 (61740 counting permuted
forms) 5-variable instances which were faces of the 6-variable
rank region.

Finally, there are the true 6-variable inequalities.  We have
found 3490 of these so far (2395095 counting permuted forms) which
pass the face verification, along with several hundred more which
do not pass and which we expect to be superseded later (though this
is not guaranteed; perhaps our stockpile of representable polymatroids
is insufficient, although the face test has been very reliable
so far).  We give some examples of these here; see the website mentioned
above for the full list.

Some inequalities follow directly from Theorem~\ref{thm:tree1},
such as:
\begin{align}
    I(A;B) &\le I(A;C)+I(B;D|C)+I(A;E|D)+I(B;F|E)+I(A;B|F) \\
    I(A;B) &\le I(A;C)+I(B;D|C)+I(A;E|D)+I(A;F|E)+I(A;B|F) \\
    I(A;B) &\le I(A;C)+I(B;D|C)+I(E;F|D)+I(A;B|E)+I(A;B|F) \\
    I(A;B) &\le I(A;C)+I(D;E|C)+I(A;B|D)+I(B;F|E)+I(A;B|F) \\
    I(A;B) &\le I(C;D)+I(A;B|C)+I(E;F|D)+I(A;B|E)+I(A;B|F)
\end{align}

And others follow directly from Theorem~\ref{thm:tree2},
such as:
\begin{align}
   2I(A;B) &\le I(A;C)+I(D;\Pair{E}{F}|C)+I(A;B|D) \notag\\
                &\Gap +I(E;F)+I(A;B|E)+I(A;B|F) \\
   2I(A;B) &\le I(A;C)+I(B;D|C)+I(A;\Pair{E}{F}|D) \notag\\
                &\Gap +I(E;F)+I(A;B|E)+I(A;B|F) \\
   2I(A;B) &\le I(C;D)+I(A;B|C)+I(B;\Pair{E}{F}|D) \notag\\
                &\Gap +I(E;F)+I(A;B|E)+I(A;B|F) \\
   2I(A;B) &\le I(\Pair{C}{D};E)+I(C;D)+I(A;F|C) \notag\\
                &\Gap +I(A;B|F)+I(A;B|D)+I(A;B|E) \\
   3I(A;B) &\le I(\Pair{C}{D};\Pair{E}{F})+I(C;D)+I(E;F)+I(A;B|C) \notag\\
                &\Gap +I(A;B|D)+I(A;B|E)+I(A;B|F)
\end{align}

Then there are inequalities which follow from Theorem~\ref{thm:tree1}
or Theorem~\ref{thm:tree2} using equivalent forms:
\begin{align}
   I(A;\Pair{B}{C}) &\le I(D;E)+I(C;F|D)+I(A;B|\Pair{D}{F}) \notag\\
                &\Gap +I(A;B|\Pair{C}{D})+I(A;C|\Pair{B}{F})+I(A;\Pair{B}{C}|E) \\
   I(\Pair{A}{B};\Pair{C}{D}) &\le I(A;\Pair{C}{D})+I(B;E|A)+I(B;D|\Triple{A}{C}{F})+I(D;F|\Pair{A}{E}) \notag\\
                &\Gap +I(B;C|\Triple{A}{E}{F})+I(B;C|\Pair{D}{E})+I(A;D|\Triple{B}{C}{F}) \notag\\
                &\Gap +I(A;C|\Triple{B}{E}{F})+I(A;F|\Triple{B}{D}{E}) \\
    2I(A;B) &\le I(D;F)+I(A;C)+I(B;D|C)+I(A;B|F)+I(A;E|D) \notag\\
                  &\Gap +I(A;F|\Pair{C}{D})+I(A;B|E)\\
    I(A;\Pair{B}{C}) &\le I(A;C)+I(B;D|C)+I(A;F|D)+I(A;B|F)+I(C;E|\Pair{B}{F}) \notag\\
                  &\Gap +I(A;C|\Pair{B}{E})\\
3I(\Pair{A}{B};\Triple{C}{D}{E}) &\le I(A;\Pair{C}{F})+I(\Pair{A}{B};D)+I(\Pair{A}{B};E)+I(C;F|D)+I(D;F|E) \notag\\
                  &\Gap +I(A;E|\Pair{D}{F})+I(B;C|\Triple{A}{D}{F})+I(B;D|\Pair{C}{F})+I(A;\Pair{D}{E}|\Pair{B}{C}) \notag\\
                  &\Gap +I(A;D|\Triple{B}{C}{E})+I(A;C|\Pair{E}{F})+I(B;D|\Triple{A}{E}{F})+I(B;\Pair{C}{D}|A) \notag\\
                  &\Gap +I(\Pair{A}{B};E|\Pair{C}{D})+I(B;E|\Triple{A}{C}{D})+I(B;D|\Triple{C}{E}{F}) \notag\\
                  &\Gap +I(\Pair{A}{B};C|\Pair{D}{E})
\end{align}
%


All of the sharp inequalities found so far using one common information
have been verified to be instances of Theorem~\ref{thm:tree2}.  It seems
quite possible that this theorem generates all one-common-information
inequalities, but we have no proof of this.

There are also hundreds of inequalities that required two common informations
to prove.  (Inequalities requiring more than two common informations
are beyond the range of our software at present.)  These are of
two types.  One type is those like inequalities (18) and (20) which
have two information terms on the left side and use the common informations
corresponding to those terms:
\begin{align}
I(A;B)+I(A;C) &\le I(B;C)+I(A;D)+I(B;E|D)+I(C;F|D) \notag\\
                  &\Gap +I(A;B|E)+I(A;C|F)\\
2I(A;\Pair{B}{C})+I(B;\Pair{C}{D}) &\le I(A;\Pair{C}{E})+I(A;F)+I(A;C|D)+2I(A;B|\Pair{C}{F}) \notag\\
                  &\Gap +I(B;C)+I(E;F|C)+2I(B;D|\Pair{C}{E})+I(C;E|F) \notag\\
                  &\Gap +I(A;D|\Pair{E}{F})+I(D;E|\Triple{A}{C}{F})+2I(A;F|\Triple{C}{D}{E})
\end{align}
The other type has just one information term on the left side but requires
a second common information in addition to the one from the left term:
\begin{align}
\label{eq:2CIa}
I(A;B) &\le I(A;C)+I(B;D|C)+I(E;F|D)+I(A;B|E)+I(A;C|F) \notag\\
                  &\Gap +I(B;E|\Pair{C}{F})\\
\label{eq:2CIb}
2I(\Pair{A}{B};\Triple{C}{D}{E}) &\le I(\Pair{A}{B};\Pair{D}{E})+I(\Triple{A}{D}{F};C)+I(\Pair{A}{F};D|C)+I(B;C|\Pair{D}{E}) \notag\\
                  &\Gap +I(A;C|B)+I(A;D|\Triple{B}{C}{E})+2I(A;C|\Triple{D}{E}{F})+I(B;C|\Triple{A}{D}{E}) \notag\\
                  &\Gap +I(A;E|\Triple{B}{D}{F})+I(B;E|\Triple{A}{C}{F})+I(B;E|\Triple{A}{D}{F})+I(B;E|\Pair{C}{D}) \notag\\
                  &\Gap +I(B;D|\Triple{A}{E}{F})+I(A;F|\Triple{B}{D}{E})+I(A;F|\Triple{B}{C}{D})\\
\label{eq:2CIc}
2I(A;\Pair{B}{C}) &\le I(A;B)+I(D;E)+I(A;B|C)+I(C;E|B)+I(D;F|\Pair{B}{E}) \notag\\
                  &\Gap +I(C;F|D)+I(A;B|\Pair{C}{D})+I(A;\Pair{B}{C}|F)+I(A;C|E)
\end{align}
Inequality~\eqref{eq:2CIa} is proved using a common information for
$A$ and $B$ along with a common information for $E$ and $(\Pair{D}{F})$;
inequality~\eqref{eq:2CIb} is proved using a common information for
$(\Pair{A}{B})$ and $(\Triple{C}{D}{E})$ along with a common information
for $(\Pair{B}{F})$ and $(\Triple{A}{D}{E})$; and
inequality~\eqref{eq:2CIc} is proved using a common information~$Z$ for
$A$ and $(\Pair{B}{C})$ along with a common information for $F$ and $Z$.
(The possible need for such iteration of common informations along with
joining of variables makes it conceivable that an unbounded number of common
informations could be needed to prove linear rank inequalities even on
a fixed number of initial variables such as~6.)

Since the inequalities in this paper have been proven using
only common informations and the Shannon inequalities, they
apply not only to linear ranks but also in any other situation
where we have random variables which are known to have
common informations.  For instance, Chan notes in
\cite[Definition~4]{Chan-ISIT07} that abelian group characterizable
random variables always have common informations (which are
still abelian group characterizable random variables); hence,
the inequalities proven here hold for such variables.

\section{An infinite list of linear rank inequalities}

The following theorem shows that there will be essentially
new inequalities for each number of variables:

\begin{theorem}
\label{thm:independence}
For any $n \ge 2$,
the inequality
\begin{equation}
\label{eq:indep}
   (n-1)I(A;B) + H(C_1C_2\dotsm C_n)
       \le \sum_{i=1}^n I(\Pair{A}{C_i};\Pair{B}{C_i})
\end{equation}
is a linear rank inequality on $n+2$ variables which is not a
consequence of instances of linear rank inequalities on
fewer than $n+2$ variables.
\end{theorem}

\begin{proof}
First, it is not hard to show that \eqref{eq:indep} is equivalent to
\eqref{eq:np2varineq}, and we have already seen that \eqref{eq:np2varineq}
is a linear rank inequality (this can also be proved using
Theorem~\ref{thm:tree2}), so \eqref{eq:indep} is a linear
rank inequality.

In the following, if $S = \{i_1,i_2,\dots,i_k\}
\subseteq \{1,2,\dots,n\}$, we will write $C_S$ for
$C_{i_1}C_{i_2}\dotsm C_{i_k}$.

Define a rank vector $v$ on the subsets of $\{A,B,C_1,C_2,\dots,C_n\}$
as follows: for any $S \subseteq \{1,2,\dots,n\}$,
\begin{align*}
      v(C_S) &= 2|S|,\\
     v(AC_S) &= n + |S|,\\
     v(BC_S) &= \min(2n - 2 + |S|,\ 2n),\\
    v(ABC_S) &= \min(2n - 1 + |S|,\ 2n).
\end{align*}
One can easily check that $v$ does not satisfy \eqref{eq:indep}.  We will
show that $v$ does satisfy all instances (using the variables
$A,B,C_1,C_2,\dots,C_n$) of all linear rank inequalities on fewer than
$n+2$ variables; this will imply that \eqref{eq:indep} is not
a consequence of these instances, as desired.

For this purpose, we construct rank vectors
$w_A,w_B,w_1,w_2,\dots,w_n$, each of which is the same as $v$ except
for one value.  The changed values are:
\begin{align*}
      w_A(A) &= n - 1,\\
      w_B(B) &= 2n - 3,\\
   w_i(BC_i) &= 2n.
\end{align*}

We will show that each of these $w$ vectors is linearly representable
over any infinite or sufficiently large finite field $F$.  In each case,
the representation will use a vector space $V$ over $F$ of dimension
$2n$, with a basis $x_1,x_2,\dots,x_n,y_1,y_2,\dots,y_n$, and the variable
$C_j$ ($1 \le j \le n$) will be represented by the two-dimensional subspace
$\langle x_j,y_j\rangle$.

For the representations of $A$ and $B$, instead of giving explicit formulas,
it will be convenient to use the following concept.  Suppose $U$ is a nontrivial
subspace of $V$.  A point $u \in U$ is said to be \textit{in general position}
in $U$, relative to a given finite set $S$ of points (if $S$ is not specified,
then we let $S$ be the set of all points that have previously been
mentioned explicitly), if $u$
does not lie in any subspace $U'$ of $V$ spanned by a subset of $S$
unless $U'$ includes all of $U$.  If the set $S$ is of size bounded by $N$, then
the ``in general position'' condition excludes at most $2^N$ proper subspaces
of $U$ (including the trivial subspace), so
there is no problem finding points in general position as long as
the field size is greater than $2^N$.  If we refer to multiple points being
chosen in general position, then they should be considered as chosen
successively, with later points being in general position relative to
earlier points as well as the previous set $S$.
This concept has been referred to by various terms; for instance, in
in~\cite{Mayhew-Newman-Whittle} such points are referred to as ``freely placed''.
Points chosen in this way make it easy to compute augmented subspace dimensions:
if $u$ is in general position in $U$ relative to $S$ and $U'$ is a subspace
spanned by points
in $S$, then $\dim(\langle U',u\rangle)$ is equal to
$\dim(U')+1$ unless $U \subseteq U'$, in which case it is
equal to $\dim(U')$.

For each $i \le n$, a representation of $w_i$ is obtained by assigning to $A$
the space $$X = \langle x_1,x_2,\dots,x_n\rangle$$ and assigning to $B$ the space
spanned by all of the $x$ vectors except $x_i$, together with $n-1$ additional
points chosen in general position in $V$.

For the representation of $w_B$, we again assign to $A$ the space
$X$; $B$ is assigned a space spanned
by $n-2$ points in general position in $X$ together with $n-1$ additional
points in general position in $V$.

To represent $w_A$, choose points $z_1,z_2,\dots,z_{n-1}$ in general
position in $X$, and assign to $A$ and $B$ the spaces
$\langle z_1,z_2,\dots,z_{n-1}\rangle$ and
$\langle z_1,z_2,\dots,z_{n-2},y_1,y_2,\dots,y_n\rangle$, respectively.

It remains to show that, if $C(t_1,\dots,t_k)\ge 0$ is a linear rank
inequality on $k$ variables with $k < n+2$, then no instance of this
inequality fails for $v$.  An instance of this inequality which applies
to $v$ is given by a map $f$ from $\{t_1,\dots,t_k\}$ to the subsets
of $\{A,B,C_1,\dots,C_n\}$.  (Then the definition of $f$ can be immediately
extended to the subsets of $\{t_1,\dots,t_k\}$ by the formula
$f(\{t_{j_1},\dots,t_{j_m}\})=f(t_{j_1})\cup\dots\cup f(t_{j_m})$.)
So suppose we have an instance, given by $C$ and $f$ as above, which
fails for $v$.  Since $C(t_1,\dots,t_k)\ge 0$ is a linear rank
inequality, the instance must not fail for the representable vector
$w_A$.  Therefore, the instance must use the value where $v$ disagrees
with $w_A$.  This means that there is a subset of $\{t_1,\dots,t_k\}$
which is mapped by $f$ to $\{A\}$; it follows that there is some
single value $j_A \in \{1,2,\dots,k\}$ such that $f(t_{j_A}) = \{A\}$.
Similarly, since
the instance must not fail for $w_B$, there is a subset of
$\{t_1,\dots,t_k\}$ which is mapped by $f$ to $\{B\}$, so there exists
$j_B \in \{1,2,\dots,k\}$ such that $f(t_{j_B}) = \{B\}$.  And, for each
$i \le n$, the instance must not fail for $w_i$, so there is a subset
of $\{t_1,\dots,t_k\}$
which is mapped by $f$ to $\{B,C_i\}$; hence, there exists
$j_i \in \{1,2,\dots,k\}$ such that $f(t_{j_i})$ is either $\{C_i\}$
or $\{B,C_i\}$.  It is clear from these $f$ values that the
numbers $j_A,j_B,j_1,j_2,\dots,j_n$ are distinct; but this is
impossible because $\{1,2,\dots,k\}$ has fewer than $n+2$ members.
This contradiction completes the proof of the theorem.
\end{proof}

\section{Concurrent work and open questions}

During the preparation of this paper, the authors became
aware of closely related concurrent work.
Chan, Grant, and Kern \cite{Chan-Grant-Kern} show nonconstructively
that there exist
linear rank inequalities not following from the Ingleton inequality.
Kinser \cite{Kinser} presents a sequence of inequalities which can
be written in the form
\begin{equation}
\label{eq:Kinser}
I(A_2;A_3) \le I(A_1;A_2) + I(A_3;A_n|A_1) + \sum_{i=4}^n I(A_2;A_{i-1}|A_i)
\end{equation}
for $n \ge 4$.  (This is a variant of \eqref{eq:starone} which follows
from Theorem~\ref{thm:list1}; the instance for $n=4$ and $n=5$ are
permuted-variable forms of the Ingleton inequality and inequality~(1c),
respectively.)  Kinser shows that \eqref{eq:Kinser} is a linear
rank inequality for each $n \ge 4$ and uses a method similar to
the proof of Theorem~\ref{thm:independence} above to show that
instance~$n$ of~\eqref{eq:Kinser} is not a consequence of linear rank
inequalities on fewer than $n$ variables.
(The authors
found the proof of Theorem~\ref{thm:independence} after the
initial posting date
of~\cite{Kinser}, but independently.)

Here are some fundamental open questions that this research has not
yet answered.

1) For each fixed $n$, are there finitely many linear rank inequalities
on $n$ variables which imply all of the others?

2) Is the method of using common informations incomplete?  That is, are
there linear rank inequalities that cannot be proved from the basic
technique of assuming the existence of common informations?

 The authors would like to thank James Oxley for helpful
discussions.

\clearpage

\renewcommand{\baselinestretch}{1.0}

\end{document}